%% file: main.tex
\RequirePackage[l2tabu,orthodox]{nag} %
\documentclass[letterpaper,11pt,]{article}

\usepackage[
  letterpaper,
  top=1in,
  bottom=1in,
  left=1in,
  right=1in
]{geometry}

\usepackage{custom}
\usepackage{macros}

\title{
  Private Edge Density Estimation for Random Graphs: \\
  Optimal, Efficient and Robust
}

\author{
  Hongjie Chen\thanks{ETH Z\"urich. This project has received funding from the European Research Council (ERC) under the European Union's Horizon 2020 research and innovation programme (grant agreement No 815464).} \and
  Jingqiu Ding\footnotemark[1] \and
  Yiding Hua\footnotemark[1] \and
  David Steurer\footnotemark[1]
}

\begin{document}

\pagestyle{empty}

\maketitle
\thispagestyle{empty} %

\begin{abstract}
  \input{abstract}
\end{abstract}
\clearpage

\microtypesetup{protrusion=false}
\tableofcontents{}
\microtypesetup{protrusion=true}
\clearpage

\pagestyle{plain}
\setcounter{page}{1}

\input{introduction}

\input{proof-sketch}

\input{sos-exp-mechanism}
\input{coarse-estimation}

\input{fine-estimation-inhomo}

\input{fine-estimation-ER}

\input{lower-bounds}

\phantomsection %
\addcontentsline{toc}{section}{References} %
\bibliographystyle{amsalpha}
\bibliography{custom.bib}

\appendix
\input{preliminaries}
\input{concentration-inequalities}

\end{document}

%% file: abstract.tex
We give the first polynomial-time, differentially node-private, and robust algorithm for estimating the edge density of \ER random graphs and their generalization, inhomogeneous random graphs.
We further prove information-theoretical lower bounds, showing that the error rate of our algorithm is optimal up to logarithmic factors.
Previous algorithms incur either exponential running time or suboptimal error rates.

Two key ingredients of our algorithm are (1) a new sum-of-squares algorithm for robust edge density estimation, and (2) the reduction from privacy to robustness based on sum-of-squares exponential mechanisms due to Hopkins et al. (STOC 2023).

%% file: introduction.tex
\section{Introduction}
\label{sec:intro}

Privacy has nowadays become a major concern in large-scale data processing.
Releasing seemingly harmless statistics of a dataset could unexpectedly leak sensitive information of individuals (see e.g. \cite{narayanan2009anonymizing, dwork2017exposed} for privacy attacks).
Differential privacy (DP) \cite{dwork2006calibrating} has emerged as a by-now standard technique for protecting the privacy of individuals with rigorous guarantees. 
An algorithm is said to be \dpt if the distribution of its output remains largely unchanged under the change of a single data point in the dataset.

For datasets represented by graphs (e.g. social networks), two notions of \dpc have been investigated in the literature: 
\edp \cite{nissim2007smooth, karwa2011private}, where each edge is regarded as a data point; and
\ndp \cite{blocki2013differentially, kasiviswanathan2013analyzing}, where each node along with its incident edges is regarded as a data point.
Node \dpc is an arguably more desirable notion than \edp.
On the other hand, \ndp is also in general more difficult to achieve without compromising on utility, as many graph statistics usually have high sensitivity in the worst case.
It turns out that many graph statistics can have significantly smaller sensitivity on typical graphs under natural distributional assumptions.
Several recent works could thus manage to achieve optimal or nearly-optimal utility guarantees in a number of random graph parameter estimation problems \cite{borgs2015private, borgs2018revealing, ullman2019efficiently, chen2024private}. 

In this paper, we continue this line of work and study perhaps the most elementary statistical task in graph data analysis:
Given an $n$-node \ER random graph of which each edge is present with probability $\pnull$ independently, 
output an estimate $\hat{p}$ of the edge density parameter $\pnull$, subject to \ndp.
We consider the error metric $\abs{\hat{p} / \pnull - 1}$ which can reflect the fact that, the task is more difficult for smaller $\pnull$.

Without privacy requirement, the empirical edge density $\hat{p}$ achieves the information theoretically optimal error rate $\abs{\hat{p} / \pnull - 1} \le \tilde{O}\paren{1 / (n \sqrt{\pnull})}$.
The standard way to achieve $\eps$-differential node privacy is to add Laplace noise with standard deviation $\Theta(1 / (\eps n))$ to the empirical edge density $\hat{p}$.
This will incur an additional privacy cost of $\Theta(1 / (\eps n \pnull))$ which dominates the non-private error $\tilde{O}\paren{1 / (n \sqrt{\pnull})}$.
Surprisingly, Borgs et al. \cite{borgs2018revealing} gave an algorithm with privacy cost only $\tilde{O}\paren{ 1 / (\eps n \sqrt{n\pnull}) }$ which is negligible to the non-private error for any $\eps \gg 1/\sqrt{n}$.
However, their algorithm is based on a general Lipschitz extension technique that has exponential running time.
Later, Sealfon and Ullman \cite{ullman2019efficiently} provided a polynomial-time algorithm based on smooth sensitivity with privacy cost $\tilde{O}\paren{ 1 / (\eps n \sqrt{n\pnull}) + 1 / (\eps^2 n^2 \pnull) }$, which is much greater than that of \cite{borgs2018revealing} for $\eps \ll 1/(\sqrt{n\pnull})$.
Moreover, \cite{ullman2019efficiently} gives evidence that their approach is inherently prohibited from achieving better privacy cost.
On the other hand, known lower bounds in \cite{borgs2018revealing, ullman2019efficiently} are not for \ER random graphs.
This leads us to the following question:
\begin{center}
    \emph{
        Is there a polynomial-time, differentially node-private, and rate-optimal edge density estimation algorithm for \ER random graphs?
    }
\end{center}

We essentially settled this question in this paper.
Specifically, we give a polynomial-time and differentially node-private algorithm with privacy cost $\tilde{O}\paren{ 1 / (\eps n \sqrt{n\pnull}) }$.
Moreover, we show this error rate is optimal up to a logarithmic factor by proving an information-theoretical lower bound of $\Omega\paren{ 1 / (\eps n \sqrt{n\pnull}) }$.
Our algorithm actually works for the more general inhomogeneous random graphs \cite{bollobas2007phase}.
The inhomogeneous random graph model encompasses any random graph model where edges appear independently (after conditioning on node labels).
Notable examples include the stochastic block model \cite{holland1983stochastic}, the latent space model \cite{hoff2002latent}, and graphon \cite{borgs2017graphons}. 

Our algorithm largely exploits the close connection between \dpc and adversarial robustness in statistics.
This connection dates back to \cite{dwork2009differential} and has witnessed significant progress in the past few years \cite{liu2021robust, liu2022differential, kothari22Private, Georgiev22RobustPrivacy, asi2023robustness, hopkins2022efficient, Hopkins2023Privacy, alabi2023privately, chen2023private, chen2024private}.
In particular, a very recent line of works \cite{hopkins2022efficient, Hopkins2023Privacy, chen2024private} could efficiently achieve optimal or nearly-optimal accuracy guarantees in a number of high-dimensional statistical tasks, by integrating two powerful tools ---\sos method \cite{raghavendra2018high} and exponential mechanisms \cite{mcsherry2007mechanism}--- in robustness and privacy respectively.
Our algorithm extends this line of work.
The key technical ingredients of our algorithm are (1) a new sum-of-squares algorithm for robust edge density estimation and (2) an exponential mechanism whose score function is based on the \sos program.
As a consequence, our private algorithm is also robust to adversarial corruptions.

\input{introduction-results}

\input{introduction-techniques}

\input{introduction-notation}

\input{introduction-organization}

%% file: introduction-results.tex
\subsection{Results}
\label{sec:results}

To state our results formally, we need the following definitions.

\begin{definition}[Node distance, neighboring graphs]
    Let $n\in \N$. 
    The node distance between two $n$-node graphs $G$ and $G'$, denoted by $\dist(G,G')$, is the minimum number of nodes in $G$ that need to be rewired to obtain $G'$.
    Moreover, we say $G$ and $G'$ are neighboring graphs if $\dist(G,G') \le 1$.
\end{definition}

\begin{definition}[Node differential privacy]
    \label{definition:node-dp}
    Let $\cG$ be the set of graphs.
    A randomized algorithm $\mathcal{A} : \cG \to \R$ is $\e$-differentially (node-)private if for every neighboring graphs $G,G'$ and every $S\subseteq\R$, we have
    \begin{equation*}
        \Pr\Brac{\mathcal{A}(G)\in S} \leq e^\eps \cdot \Pr\Brac{\mathcal{A}(G')\in S} \,.
    \end{equation*}
\end{definition}

\begin{definition}[Node corruption model]
    Let $n \in \N$ and $\eta \in [0,1]$.
    For an $n$-node graph $G$, we say an $n$-node graph $G'$ is an $\eta$-corrupted version of $G$ if $\dist(G,G') \le \eta n$.
\end{definition}

\paragraph{\ER random graphs}

We provide a polynomial-time,  differentially node-private and robust edge density estimation algorithm for \ER random graphs.

\begin{theorem}[\ER random graphs, combination of \cref{thm:coarse_estimation_inhomo} and \cref{thm:fine_estimation_er}]
    \label{thm:algorithm_ER}
    There are constants $C_1, C_2, C_3$ such that the following holds.
    For any $\eta\le C_1$, $\eps \ge C_2\log(n)/n$, and $\pnull\geq C_3/n$, there exists a polynomial-time $\eps$-differentially node-private algorithm which, given an $\eta$-corrupted \ER random graph $\bbG(n,\pnull)$, outputs an estimate $\tilde{p}$ satisfying
    \[
        \Abs{\frac{\tilde p}{\pnull}-1} 
        \le
        O\Paren{
            \frac{\sqrt{\log n}}{ n \sqrt{\pnull} } +
            \frac{\log^2 n}{\eps n \sqrt{n\pnull}} +
            \frac{\eta\log n}{\sqrt{n\pnull}} 
        } \,,
    \]
    with probability $1-n^{-\Omega(1)}$.
\end{theorem}

The first term $O(\sqrt{\log n} / (n \sqrt{\pnull}))$ is the sampling error that is necessary even without privacy or robustness.
The second term $O(\log^2(n) / (\eps n \sqrt{n\pnull}))$ is the privacy cost of our algorithm, which matches the exponential-time algorithm in \cite{borgs2018revealing}.
The third term $O( {\eta \log n} / {\sqrt{n\pnull}} )$ is the robustness cost of our algorithm, which matches the information-theoretical lower bound $\Omega( \eta / \sqrt{n\pnull}  )$ in \cite[Theorem 1.5]{acharya2022robust} up to a $\log n$ factor.

Moreover, we provide the following lower bound which shows that the privacy cost of our algorithm is optimal up to a $\log n$ factor.\footnote{
    Borgs et al. \cite{borgs2018revealing} proved a lower bound for a variant of \ER random graphs. However, it is not clear whether their proof technique can be easily extended to \ER random graphs.
}

\begin{theorem}[Privacy lower bound for \ER random graphs]
    \torestate{
        \label{thm:lower_bound_ER}
        Suppose there is an $\eps$-differentially node-private algorithm that, given an \ER random graph $\bbG(n,\pnull)$, outputs an estimate $\tilde{p}$ satisfying $\abs{\tilde{p} / \pnull - 1} \le \alpha$ with probability $1-\beta$.
        Then we must have
        \[
            \alpha 
            \geq 
            \Omega\Paren{ 
                \frac {\log(1/\beta)} {\eps n \sqrt{n\pnull}}
            } \,.
        \]
    }
\end{theorem}

\paragraph{Inhomogeneous random graphs}
Given an $n$-by-$n$ edge connection probability matrix $\Qnull$, the inhomogeneous random graph model $\bbG(n, \Qnull)$ defines a distribution over $n$-node graphs where each edge $\{i,j\}$ is present with probability $(\Qnull)_{ij}$ independently.

We provide a polynomial-time, differentially node-private 
and robust edge density estimation algorithm for inhomogeneous random graphs.

\begin{theorem}[Inhomogeneous random graphs, combination of \cref{thm:coarse_estimation_inhomo} and \cref{thm:fine_estimation_inhomo}]
    \label{thm:algorithm_inhomo}
    Let $\Qnull$ be an $n$-by-$n$ edge connection probability matrix and let $\pnull \seteq \sum_{i,j} \Qnull_{ij} / (n^2-n)$. 
    Suppose $\Norm{\Qnull}_{\infty} \le R\pnull$ for some $R$.
    There is a sufficiently small constant $c$ such that the following holds.
    For any $\eta$ such that $\eta\log(1/\eta)R \le c$, there exists a polynomial-time $\eps$-differentially node-private algorithm which, given an $\eta$-corrupted inhomogeneous random graph $\bbG(n,\Qnull)$, outputs an estimate $\tilde{p}$ satisfying
    \[
        \Abs{\frac{\tilde{p}}{\pnull}-1} 
        \le
        O\Paren{
            \frac{\sqrt{\log n}}{ n \sqrt{\pnull} } +
            \frac{R \log^2 n}{\eps n} +
            R \eta \log(1/\eta)
        } \,,
    \]
    with probability $1-n^{-\Omega(1)}$.
\end{theorem}

To the best of our knowledge, even without privacy requirement and in the special case of \ER random graphs, no previous algorithm can match our guarantees in the sparse regime.
Specifically, when $\dnull \ll \log n$ and $\eta\ge\Omega(1)$, our algorithm can provide a constant-factor approximation of $\dnull$, while the best previous robust algorithm \cite{acharya2022robust} can not.

We also provide matching lower bounds, showing that the guarantee of our algorithm in \cref{thm:algorithm_inhomo} is optimal up to logarithmic factors.

\begin{theorem}[Robustness lower bound for inhomogeneous random graphs]
    \torestate{
        \label{thm:robust_lower_bound_inhomo}
        Suppose there is an algorithm satisfies the following guarantee for any symmetric matrix $\Qnull \in [0,1]^{n \times n}$.
        Given an $\eta$-corrupted inhomogeneous random graph $\bbG(n, \Qnull)$, the algorithm outputs an estimate $\hat{p}$ satisfying $\abs{\hat{p}/\pnull - 1} \le \alpha$ with probability at least $0.99$, where $\pnull = \sum_{i,j} \Qnull_{ij} / (n^2-n)$.
        Then we must have $\alpha \geq \Omega(R\eta)$, where $R = \max_{i,j} \Qnull_{ij} / \pnull$.
    }
\end{theorem}

\begin{theorem}[Privacy lower bound for inhomogeneous random graphs]
    \torestate{
        \label{thm:privacy_lower_bound_inhomo}
        Suppose there is an $\eps$-differentially node-private algorithm satisfies the following guarantee for any symmetric matrix $\Qnull \in [0,1]^{n \times n}$.
        Given an inhomogeneous random graph $\bbG(n, \Qnull)$, the algorithm outputs an estimate $\hat{p}$ satisfying $\abs{\hat{p} / \pnull - 1} \le \alpha$ with probability $1-\beta$, where $\pnull = \sum_{i,j} \Qnull_{ij} / (n^2-n)$.
        Then we must have 
        \[
            \alpha \geq \Omega\Paren{\frac{R\log(1/\beta)}{n\eps} } \,,
        \]
        where $R = \max_{i,j} \Qnull_{ij} / \pnull$.
    }
\end{theorem}

%% file: introduction-techniques.tex
\subsection{Techniques}
\label{sec:techniques}

We give an overview of the key techniques used to obtain our algorithm.
As our techniques for \ER random graphs can be easily extended to the more general inhomogeneous random graph model, we will focus on \ER random graphs to avoid a proliferation of notation.
Specifically, given an $\eta$-corrupted \ER random graph $\bbG(n,\dnull/n)$, our goal is to output a private estimate of $\dnull$.

\paragraph{Reduction from privacy to robustness}
Hopkins et al. \cite{Hopkins2023Privacy} and Asi et al. \cite{asi2023robustness} independently discovered the following black-box reduction from privacy to robustness. 
Given a robust algorithm $\cA_{\mathrm{robust}}$, one can directly obtain a private algorithm via applying the exponential mechanism \cite{mcsherry2007mechanism} with the following score function,
\begin{equation}
    \label{eq:black_box_reduction}
    \score(d; A) \seteq \min_{A'} \Set{ \dist(A', A) \,:\, \Abs{ \cA_{\mathrm{robust}}(A') - d } \le 1/\poly(n) } \,,
\end{equation}
where $A$ is the adjacency matrix of input graph and $d$ is a candidate estimate.
For privacy analysis, note that the sensitivity of the above score function is bounded by $1$, as the node distance between neighboring graphs is at most $1$.
For utility analysis, when the input graph is a typical \ER random graph, the exponential mechanism will with high probability output a $\hat{d}$ of score $O(\log(n)/\e)$.
Then we can argue that such a $\hat{d}$ is close to $\dnull$ using the robustness of $\cA_{\mathrm{robust}}$.
For example, if we plug in the robust algorithm in \cite[Theorem 1.3]{acharya2022robust}, then the corresponding exponential mechanism will only incur a privacy cost of $\tilde{O} \paren{1 / (\epsilon n\sqrt{\dnull})}$.

However, directly plugging in the robust algorithm in \cite{acharya2022robust} will lead to an exponential-time algorithm, as a single evaluation of the score function requires enumerating all $n$-node graphs.
To obtain a polynomial-time algorithm, we develop a new robust algorithm via the \emph{sum-of-squares} method.

\paragraph{Robust algorithm via sum-of-squares}
The \sos method uses convex programming (in particular, semidefinite programming) to solve polynomial programming.
It is a very powerful tool for designing polynomial-time robust estimators (see \cite{raghavendra2018high}).
To obtain a robust algorithm via sum-of-squares, we first identify a set of polynomial constraints that a typical (uncorrupted) \ER random graph would satisfy. 
Specifically, these polynomial constraints encode the following regularity conditions:
(1) the degrees of the nodes are highly concentrated, and (2) the centered adjacency matrix is spectrally bounded.
We also include the constraint that at most $\eta$ fraction of the nodes in the graph are corrupted.
Then we give a proof that if a graph satisfies the above constraints, then its average degree will be close to $\dnull$, even when $\eta$ fraction of nodes in the input graph are arbitrarily corrupted.
Importantly, the proof is simple enough that it is captured by the \sos proof system (see \cite{fleming2019semialgebraic}).
This allows us to extend the utility guarantee of the polynomial program to its semidefinite programming relaxation, which results in a polynomial-time robust algorithm.

\paragraph{Sum-of-squares exponential mechanism}
Given the above robust algorithm, we then use the sum-of-squares exponential mechanism developed in \cite{hopkins2022efficient,Hopkins2023Privacy} to obtain a private algorithm.
More specifically, we apply the exponential mechanism with the \sos relaxation of the score function in \cref{eq:black_box_reduction}.
In this way, we obtain a private algorithm that is also robust to adversarial corruptions.

%% file: introduction-notation.tex
\subsection{Notation}
\label{sec:notation}

We introduce some notation used throughout this paper.
We write \(f \lesssim g\) to denote the inequality \(f \le C \cdot g\) for some absolute constant \(C>0\).
We write \(O(f)\) and \(\Omega(f)\) to denote quantities \(f_-\) and \(f_+\) satisfying \(f_-\lesssim f\) and \(f \lesssim f_+\) respectively.
We use boldface to denote random variables, e.g., $\bm X, \bm Y, \bm Z$.
For a matrix $M$, we use $\normop{M}$ for the spectral norm of $M$.
Let $\one$ and $\zero$ denote the all-one and all-zero vector respectively, of which the size will be clear from the context.
We use a graph $G$ and its adjacency matrix $A=A(G)$ interchangeably when there is no ambiguity.
For an $n$-by-$n$ matrix $M$, we use $d(M)$ to denote its average row/column sum, i.e., $d(M) = \sum_{i,j} M_{ij}/n$.
For any matrices (or vectors) $M,N$ of the same shape, we use $M\odot N$ to denote the element-wise product (aka Hadamard product) of $M$ and $N$.

%% file: introduction-organization.tex
\subsection{Organization}
\label{sec:organization}

The rest of the paper is organized as follows.
In \cref{sec:proof-sketch}, we give a proof overview of our results.
In \cref{sec:sos_exp_mechanism}, we present a general \sos exponential mechanism that all of our private algorithms in this paper are based on.
In \cref{sec:coarse_estimation}, we present our coarse estimation algorithm and give a full proof of its guarantees (\cref{thm:coarse_estimation_inhomo}).
In \cref{sec:fine_estimation_inhmomo}, we present our fine estimation algorithm for inhomogeneous random graphs and give a full proof of its guarantees (\cref{thm:fine_estimation_inhomo}).
In \cref{sec:fine_estimation_ER}, we present our fine estimation algorithm for \ER random graphs and give a full proof of its guarantees (\cref{thm:fine_estimation_er}).
All lower bounds are proved in \cref{sec:lower-bounds}.
We provide some \sos background in \cref{sec:prelim} and some concentration inequalities for random graphs in \cref{sec:concentration-inequalities}.

%% file: proof-sketch.tex
\section{Private and robust algorithm for \ER random graphs}
\label{sec:proof-sketch}

In this section, we describe our private and robust algorithm for \ER random graphs.
We also give an overview of the analysis of our algorithm and sketch the proof of our lower bounds.

Our overall algorithm consists of two stages. 
In the first stage, we compute a coarse estimate that approximates the edge density parameter within constant factors. 
In the second stage, we improve the accuracy of this coarse estimate to the optimum.
Since our algorithm is private in both stages, it is also private overall by the composition theorem of differential privacy (see \cite[Section 3.5]{dwork2014algorithmic}).

We remark that for the \ER random graph model $\bbG(n, \pnull)$, estimating its edge density parameter $\pnull$ is equivalent to estimating its expected average degree $\dnull \seteq n\pnull$.
For the convenience of notation, we set our goal as estimating the expected average degree $\dnull$ throughout this section.

\input{proof-sketch-general-algorithm-framework}

\input{proof-sketch-coarse-estimation}

\input{proof-sketch-fine-estimation}

\input{proof-sketch-lower-bound}

%% file: proof-sketch-general-algorithm-framework.tex
\subsection{General algorithm framework}

Given an $n$-by-$n$ symmetric matrix $A$ and a scalar $\gamma\in[0,1]$, 
let $\cT(Y, z; A, \gamma)$ be a polynomial system with indeterminates $Y=(Y_{ij})_{i,j\in[n]}$ and $z=(z_i)_{i\in[n]}$ that encodes the node distance between $Y$ and $A$:
\begin{equation}\label{eq:techniques-graph-constraints}
    \cT(Y,z; A,\gamma) \seteq 
    \left. 
        \begin{cases}
            z \odot z = z,\, \iprod{\one, z} \ge (1-\gamma)n \\
            0 \leq Y \leq \one \one^{\top},\, Y=Y^\top \\
            Y \odot zz^\top = A \odot zz^\top
        \end{cases} 
    \right\} \,.
\end{equation}
Let $\cR(Y)$ be A polynomial system that encodes regularity conditions of \ER random graphs. The key observation here is that, for any $Y \in \set{0, 1}^{n \times n}$ and $z \in \set{0, 1}^{n}$ that satisfy constraints in $\cT(Y, z; A, \gamma) \cup \cR(Y)$, $Y$ is a graph that behaves like \ER random graphs (in the sense of the regularity conditions) and is within node distance $\gamma n$ to $A$ where they agree on $\set{i\in[n] \,:\, z_i=1}$.

The key ingredient of our result is that, given proper regularity conditions $\cR(Y)$, we can give degree-$8$ sum-of-squares proofs: for any $Y$ that satisfies constraints in $\cT(Y, z; A, \gamma) \cup \cR(Y)$, the average degree of $Y$ is close to the expected average degree $\dnull$, even when the input graph $A$ is a $\gamma$-corrupted \ER random graph $\bbG(n, \dnull/n)$.
As a result of the sum-of-squares proofs-to-algorithms framework (see \cref{theorem:SOS_algorithm}), we can get an efficient and robust estimator $\Tilde{\E}[d(Y)]$, where $\Tilde{\E}$ is a pseudo-expectation obtained by solving level-$8$ sum-of-squares relaxation of $\cT(Y, z; A, \gamma) \cup \cR(Y)$.

Based on the above identifiability proof for robust estimation, we design a private and robust algorithm by applying the exponential mechanism\footnote{
    To efficiently implement this exponential mechanism, we note that the score function \cref{eq:sos_relaxation} can be evaluated in polynomial time by combining binary search and semidefinite programming. 
    By discretizing $[0,n]$ with step size $1/\poly(n)$, one can sample from the distribution \cref{eq:sos_exp_mec} with a polynomial number of queries to the score function. 
    For more detailed discussions, see \cref{remark:score_function_computation} and \cref{remark:sampling}.
}
with the following score function:
\begin{align}
    \label{eq:sos_relaxation}
    \begin{split}
        \soscore(d; A) \seteq 
        \min_{0\le\gamma\le1} \gamma n \text{ s.t. }
            & \exists \text{ level-}8 \text{ pseudo-expectation } \tE \text{ satisfying } \\[-8pt]
            & \cT(Y,z; A,\gamma) \,\cup\, 
            \cR(Y) \,\cup\, 
             \Set{\Abs{d(Y)-d} \le 1/\poly(n)}  \,.
    \end{split}
\end{align}

Similar to \cref{eq:black_box_reduction}, it is easy to observe this exponential mechanism is private.

\begin{lemma}[Privacy]
    \label{fact:sketch-privacy}
    Consider the distribution $\mu_{A,\eps}$ with support $[0,n]$ and density 
    \begin{equation}
        \label{eq:sos_exp_mec}
        \mathrm{d}\mu_{A,\eps}(d) \propto \exp\paren{-\eps\cdot \soscore(d;A)} \,,
    \end{equation}
    where $\soscore(d;A)$ is defined in \cref{eq:sos_relaxation}.
    A sample from $\mu_{A,\eps}$ is $2\e$-differentially private. 
\end{lemma}
\begin{proof}
    Since the node distance between neighboring graphs is at most $1$, the sensitivity of the following score function is bounded by $1$:
   \begin{equation*}
        \score(d;A) \seteq 
        \min_{0\le\gamma\le1} \gamma n \text{ s.t. } 
        \cT(Y,z; A,\gamma) \,\cup\, 
        \cR(Y) \,\cup\, 
        \Set{\Abs{d(Y)-d} \le 1/\poly(n)} 
        \text{ is feasible.} 
    \end{equation*}
    One can show that such sensitivity bound is inherited by its sum-of-squares relaxation $\soscore$ as defined in \cref{eq:sos_relaxation}.
    By a standard sensitivity-to-privacy argument (see e.g. \cite[Theorem 3.10]{dwork2014algorithmic}), the exponential mechanism is $2\e$-differentially private.
\end{proof}

To analyze the utility of the private algorithm, we use the robustness of the score function. 
Assume the input graph is uncorrupted for simplicity.
For a typical \ER random graph $\Anull\sim \bbG(n,\dnull/n)$, we have $\soscore(\dnull,\Anull)=0$.
By a standard volume argument (see e.g. \cite[Theorem 3.11]{dwork2014algorithmic}), the exponential mechanism with high probability outputs a scalar $d$ satisfying $\soscore(d; \Anull) \lesssim \log(n)/\eps$.
By the definition of our score function  in \cref{eq:sos_relaxation}, this implies that there exists a level-8 pseudo-distribution satisfying $\cT(Y,z; \Anull,\gamma) \,\cup\, \cR(Y)$ with $\gamma \lesssim \log(n)/ (\e n)$. 
The utility then follows from the above identifiability proof for robust estimation.

%% file: proof-sketch-coarse-estimation.tex
\subsection{Coarse estimation}
\label{sec:proof-sketch-coarse-estimation}

In this part, we describe a private and robust algorithm that can estimate the expected average degree $\dnull$ within a constant approximation ratio. 

\begin{theorem}[Coarse estimation algorithm, informal restatement of \cref{thm:coarse_estimation_inhomo}]
    \label{thm:sketch-coarse-estimation}
    For $\eta$ smaller than some constant, there is a polynomial-time $\e$-differentially node-private algorithm which, given an $\eta$-corrupted \ER random graph $\bbG(n,{\dnull}/{n})$, outputs an estimate $\hat{d}$ such that $\abs{\hat{d}-\dnull}\leq 0.5 \dnull$.
\end{theorem}

We give a proof sketch of \cref{thm:sketch-coarse-estimation} at the end of this subsection.
The formal theorem and proofs are deferred to \cref{sec:coarse_estimation}.

\paragraph{Identifiability proof for robust estimation}  
We first give a polynomial system that can identify the expected average degree $\dnull$ up to constant factors, even when $\eta$-fraction of nodes are corrupted. 
Consider the following regularity condition on degrees:
\begin{equation}
    \label{eq:regularity_coarse}
    \cR(Y) \seteq \Set{ (Y\one)_i \le 2\log(1/\eta) \cdot d(Y) \,, \quad \forall i\in[n] } \,.
\end{equation}
The following lemma shows that \ER random graphs satisfy $\cT(Y,z; A,2\eta) \cup \cR(Y)$ with high probability.
\begin{lemma}[Feasibility]
    \label{lem:sketch-feasibility}
    Let $A^{\circ} \sim \bbG(n,{\dnull}/{n})$ and let $A$ be an $\eta$-corrupted version of $\Anull$.
    With high probability, there exists a graph $Y$ that satisfies the constraints in $\cT(Y,z; A,2\eta) \cup \cR(Y)$.
\end{lemma}
\begin{proof}[Proof sketch]
    For $\dnull \gg \log(n)$, the maximum degree of $\Anull$ is of order $O(\dnull)$. 
    Therefore, the uncorrupted graph $A^{\circ}$ satisfies the constraints.
    For $\dnull \ll \log n$, using concentration properties of random graphs, we can show that the number of high degree nodes is bounded by $\eta n$.
    A feasible graph can then be obtained from the uncorrupted graph $A^{\circ}$ by trimming these highest degree nodes.
\end{proof}

Next, we show that these polynomial constraints give an identifiability proof for the expected average degree $\dnull$.
\begin{lemma}[Identifiability]
    \label{lem:sketch-coarse-identifiability}
    Let $A^{\circ} \sim \bbG(n,{\dnull}/{n})$ and let $A$ be an $\eta$-corrupted version of $\Anull$.
    For $\eta$ smaller than some constant and $\gamma\leq O(\eta)$, 
    with high probability there is a degree-$8$ sum-of-squares proof that,
    if $Y$ satisfies $\cT(Y,z; A,\gamma) \cup \cR(Y)$,
    then $\abs{d(Y)-\dnull}\leq 0.001 \dnull$.
\end{lemma}
\begin{proof}[Proof sketch]
    We first assume that $\dnull \gg \log(n)$, for which the proof is simpler. 
    By the degree-bound constraint $\cR(Y)$, we have $n\abs{d(Y)-d(A^\circ)} \leq 2\log(1/\eta) \cdot (d(Y)+\dnull)\cdot \text{dist}(Y,A^\circ)$.
    Using the constraints $Y\odot zz^\top= A\odot zz^\top$ and $\iprod{1,z}\ge(1-\gamma)n$, we have $\dist(Y,A) \le \gamma n$.
    Since $\dist(A,\Anull) \le \eta n$, by triangle inequality, we have $\dist(Y, \Anull) \leq (\gamma+\eta)n$.
    Therefore, we have $\abs{d(Y)-d(A^\circ)}\leq 0.0001 \dnull$ when $\gamma,\eta$ are at most some small constants. 
    Finally, by random graph concentration, we have $\abs{\dnull-d(A^\circ)}\leq o(\dnull)$ with high probability.
    Therefore, we have $\abs{d(Y)-d(A^\circ)}\leq 0.001 \dnull$.
    
    To deal with the sparse regime where $\dnull \ll \log n$, we need to truncate the nodes of $A^\circ$ with degree $\Omega(\log(1/\eta) \dnull)$. 
    Our key observation is that, the average degree of the graph before and after truncation only differ by a constant factor. 
    Therefore, we can still get $\abs{d(Y)-d(A^\circ)}\leq 0.001 \dnull$.

    Furthermore, it can be shown that this proof is a degree-8 sum-of-squares proof.
\end{proof}

\paragraph{Robust algorithm via sum-of-squares}
Consider the algorithm that finds a level-$8$ pseudo-expectation satisfying $\cT(Y,z; A,2\eta) \,\cup\, \cR(Y)$ ---with $\cR(Y)$ given in \cref{eq:regularity_coarse}--- and outputs $\tE [d(Y)]$.
By \cref{lem:sketch-feasibility}, such a pseudo-expectation $\tE$ exists with high probability. 
It follows from the \sos identifiability proof in \cref{lem:sketch-coarse-identifiability} that $\abs{\tE [d(Y)]-\dnull}\leq 0.001\dnull$.
Moreover, the algorithm can be implemented by semidefinite programming and run in polynomial time.

\paragraph{Private and robust algorithm via sum-of-squares exponential mechanism} 
We present our private and robust algorithm in \cref{alg:coarse_sos_exp_mec} and give a proof sketch of \cref{thm:sketch-coarse-estimation}.

\begin{algorithmbox}[Private coarse estimation for \ER random graphs]
    \label{alg:coarse_sos_exp_mec}
    \mbox{}\\
    \textbf{Input:} $\eta$-corrupted \ER random graph $A$.
    
    \noindent
    \textbf{Privacy parameter:} $\eps$.
    
    \noindent
    \textbf{Output:} 
        A sample from the distribution $\mu_{A,\eps}$ with support $[0,n]$ and density 
        \begin{equation}
            \label{eq:coarse_sos_exp_mec}
            \mathrm{d}\mu_{A,\eps}(d) \propto \exp\paren{-\eps\cdot \soscore(d;A)} \,,
        \end{equation}
        where 
        \begin{equation}
            \label{eq:coarse_sos_exp_mec_score}
         \begin{split}
    \soscore(d; A) \seteq 
    \min_{0\le\gamma\le1} \gamma n \text{ s.t. }
        & \exists \text{ level-}8 \text{ pseudo-expectation } \tE \text{ satisfying } \\[-8pt]
        & \cT(Y,z; A,\gamma) \,\cup\, 
        \cR(Y) \,\cup\, 
        \Set{\Abs{d(Y)-d} \le 1/\poly(n)}  \,,
\end{split}
        \end{equation}
        with $\cR(Y)$ given in \cref{eq:regularity_coarse}.
\end{algorithmbox}

\begin{proof}[Proof sketch of \cref{thm:sketch-coarse-estimation}]
    \emph{Privacy.}
    By \cref{fact:sketch-privacy}, \cref{alg:coarse_sos_exp_mec} is $2\e$-differentially private. 
    
    \emph{Utility.} 
    For simplicity, we consider the case when there is no corruption (i.e. $\eta=0$).
    The analysis for the case when $\eta>0$ is similar.
    Let $\Anull \sim \bbG(n, \dnull/n)$.
    Then with high probability $\soscore(\dnull;\Anull)=0$.
    By a standard volume argument, \cref{alg:coarse_sos_exp_mec} outputs a scalar $d$ that satisfies $\soscore(d; \Anull) \lesssim \log(n) / \eps$ with high probability.
    By the definition of $\soscore$ in \cref{eq:coarse_sos_exp_mec_score}, this implies that there exists a level-$8$ pseudo-distribution satisfying $\cT(Y,z; A,\gamma) \,\cup\, \cR(Y) \,\cup\, \set{\abs{d(Y)-d} \le 1/\poly(n)}$ with $\gamma \lesssim \log(n)/(\e n)$. 
    When $\log(n) / (\e n)$ is at most a small constant, it follows from our \sos identifiability proof in \cref{lem:sketch-coarse-identifiability} that, \cref{alg:coarse_sos_exp_mec} outputs a constant-factor approximation of $\dnull$ with high probability.
\end{proof}

%% file: proof-sketch-fine-estimation.tex
\subsection{Fine estimation}
\label{sec:proof-sketch-fine-estimation}

From \cref{sec:proof-sketch-coarse-estimation}, we know how to obtain a constant-factor approximation of $\dnull$ privately and robustly.
In this section, we show how to improve the accuracy to the optimum.

\begin{theorem}[Fine estimation algorithm, informal restatement of \cref{thm:fine_estimation_er}]
    \label{thm:sketch-fine-estimation}
    Let $0.5\dnull \le \hat{d} \le 2\dnull$.
    For $\eta$ smaller than some constant, 
    there is a polynomial-time $\eps$-differentially node-private algorithm which,
    given an $\eta$-corrupted \ER random graph $\bbG(n,\dnull/n)$ and $\hat{d}$,
    outputs an estimate $\tilde d$ such that 
    \begin{equation*}
        \Abs{\frac{\tilde{d}}{\dnull}-1}\leq \tilde{O}\Paren{\frac{1}{\sqrt{n\dnull}}+\frac{1}{\e n\sqrt{\dnull}}+\frac{\eta}{\sqrt{\dnull}}}\,.
    \end{equation*}
\end{theorem}

We give a proof sketch of \cref{thm:sketch-fine-estimation} at the end of this section. The formal theorem and proofs are deferred to \cref{sec:fine_estimation_ER}.

\paragraph{Identifiability proof for robust estimation} 
We first give a polynomial system which can identify the expected average degree $\dnull$ with optimal error rate, when provided with a coarse estimate $\hat{d}$.
Consider the following regularity conditions on degrees and eigenvalues:
\begin{equation}
    \label{eq:regularity_fine}
    \cR(Y) \seteq
    \left. 
        \begin{cases}
            \Abs{(Y\one)_i - d(Y)} \le \sqrt{\hat{d}} \log n \,, & \forall i\in[n] \\
            \Normop{Y - \frac{d(Y)}{n} \one \one^{\top}} \leq \sqrt{\hat{d} \log n}
        \end{cases}
    \right\} \,.
\end{equation}

\begin{lemma}[Feasibility]
\label{lem:sketch-fine-feasibility}
    Let $A^{\circ} \sim \bbG(n,{\dnull}/{n})$ and let $A$ be an $\eta$-corrupted version of $\Anull$.
    Suppose $\dnull/2\leq\hat{d}\leq 2\dnull$.
    Then with high probability, there exists a graph $Y$ that satisfies the constraints in $\cT(Y,z; A,\eta) \cup \cR(Y)$.
\end{lemma}
\begin{proof}
    By Chernoff bound, with high probability, the degree of each node in $\Anull$ deviates from $\dnull$ by at most $O\paren{ \sqrt{\dnull} \log n }$.
    By the concentration of the spectral norm of random matrices~\cite{bandeira2016sharp}, with high probability, we have $\normop{\Anull - \frac{d(\Anull)}{n} \one\one^\top} \lesssim {\sqrt{\dnull \log n}}$.
    Hence, $\cT(Y,z; A,\eta) \cup \cR(Y)$ is satisfied by $Y=\Anull$ and $z=z^{\circ}$ where $z^{\circ}$ is the indicator vector for uncorrupted nodes.
\end{proof}

Next we give a sum-of-squares identifiability proof for expected average degree estimation with optimal accuracy.
\begin{lemma}[Identifiability]
    \label{lem:sketch-fine-identifiability}
    Let $A^{\circ} \sim \bbG(n,{\dnull}/{n})$ and let $A$ be an $\eta$-corrupted version of $\Anull$.
    Suppose $\dnull/2\leq\hat{d}\leq 2\dnull$.
    For $\eta$ smaller than some constant and $\gamma\leq O(\eta)$, 
    with high probability there is a degree-$8$ sum-of-squares proof that,
    if $Y$ satisfies $\cT(Y,z; A,\eta) \cup \cR(Y)$, then
    \begin{equation*}
        \Abs{\frac{d(Y)}{\dnull}-1}\leq \tilde{O}\Paren{\frac{1}{\sqrt{n\dnull}}+\frac{\eta}{\sqrt{\dnull}}}\,.
    \end{equation*}
\end{lemma}

\begin{proof}[Proof sketch]
Let $Y_1,Y_2$ be two graphs satisfying the regularity condition $\cR(Y_1)$ and $\cR(Y_2)$ as described in \cref{eq:regularity_fine}, respectively.
We give sum-of-squares proof that, if $\dist(Y_1, Y_2) \le \zeta n$ and $\zeta$ is at most some small constant, then \(\Abs{ d(Y_1) - d(Y_2) } \le \zeta \sqrt{\hat d} \log n \)\,.

Let $w\in \Set{0,1}^n$ be the indicator vector for the shared induced subgraph between $Y_1$ and $Y_2$, i.e $Y_1\odot ww^\top=Y_2\odot ww^\top$.
When $\dist(Y_1, Y_2) \le \zeta n$, we have $\iprod{w,\one}\geq (1-\zeta)n$.
We have
\begin{align*}
\begin{split}
    n \Bigparen{d(Y_1) - d(Y_2)} = & \iprod{Y_1-Y_2,\one \one^{\top}} \\
    = & \iprod{Y_1-Y_2,\one \one^{\top} - ww^{\top}} \\
    = & \iprod{Y_1-\frac{d(Y_1)}{n}\one \one^{\top} +\frac{d(Y_1)}{n}\one \one^{\top}-\frac{d(Y_2)}{n}\one \one^{\top}+\frac{d(Y_2)}{n}\one \one^{\top}-Y_2,\one \one^{\top} - ww^{\top}} \\
    = & \iprod{Y_1-\frac{d(Y_1)}{n}\one \one^{\top},\one \one^{\top} - ww^{\top}} + \iprod{\frac{d(Y_2)}{n}\one \one^{\top}-Y_2,\one \one^{\top} - ww^{\top}} \\
    & + \iprod{\frac{d(Y_1)}{n}\one \one^{\top}-\frac{d(Y_2)}{n}\one \one^{\top},\one \one^{\top} - ww^{\top}} \,.
\end{split}
\end{align*}
By rearranging terms, we can get
\begin{equation*}
    \frac{\iprod{\one, w}^2}{n} \Bigparen{d(Y_1) - d(Y_2)}
    = \iprod{Y_1-\frac{d(Y_1)}{n}\one \one^{\top},\one \one^{\top} - ww^{\top}} + \iprod{\frac{d(Y_2)}{n}\one \one^{\top}-Y_2,\one \one^{\top} - ww^{\top}} \,.
\end{equation*}
For the first term $\iprod{Y_1-\frac{d(Y_1)}{n}\one \one^{\top},\one \one^{\top} - ww^{\top}}$, we have
\begin{equation*}
    \iprod{Y_1-\frac{d(Y_1)}{n}\one \one^{\top},\one \one^{\top} - ww^{\top}}
    = 2 \iprod{Y_1-\frac{d(Y_1)}{n}\one \one^{\top},\one (\one - w)^{\top}} + \iprod{\frac{d(Y_1)}{n}\one \one^{\top}-Y_1,(\one - w) (\one - w)^{\top}} \,.
\end{equation*}
From constraints $|(Y_1\one)_i - d(Y_1)| \le  \sqrt{\hat{d}}\log(n)$ for all  $i\in[n]$,  we have
\begin{align*}
    \iprod{Y_1-\frac{d(Y_1)}{n}\one \one^{\top},\one (\one - w)^{\top}}
     = \iprod{Y_1 \one-d(Y_1)\one,\one - w} 
     \leq \zeta n \log(n) \sqrt{\hat{d}} \,.
\end{align*}
From constraints $\Normop{Y_1 - \frac{d(Y_1)}{n} \one \one^{\top}} \leq \delta \sqrt{\hat{d}}$, we have
\begin{align*}
    \iprod{\frac{d(Y_1)}{n}\one \one^{\top}-Y_1,(\one - w) (\one - w)^{\top}}
    \leq \Normop{Y_1 - \frac{d(Y_1)}{n} \one \one^{\top}} \Norm{\one - w}_2^2 
    \leq \zeta n  \sqrt{\hat{d}} \log(n)\,,
\end{align*}

The same bounds also apply for the second term $\iprod{Y_2-\frac{d(Y_2)}{n}\one \one^{\top},\one \one^{\top} - ww^{\top}}$. 
Since $\iprod{\one,w}\geq \Omega(n)$, it follows that
$\Abs{d(Y_1) - d(Y_2)}  \leq \tilde{O}(\zeta \sqrt{\hat{d}}) \leq \tilde{O}(\zeta  \sqrt{\dnull})$.

Since the original uncorrupted graph satisfies the regularity conditions, this gives the identifiability proof that $\abs{d(Y)-d(\Anull)}\leq \tilde{O}(\zeta  \sqrt{\dnull})$.
By random graph concentration, with high probability, we have $\abs{\dnull-d(\Anull)}\leq \tilde{O}(\sqrt{\dnull/n})$. 
The claim thus follows.
\end{proof}

\paragraph{Robust algorithm via sum-of-squares}
Consider the algorithm that finds a level-$8$ pseudo-expectation satisfying $\cT(Y,z; A,\eta) \,\cup\, \cR(Y)$ ---with $\cR(Y)$ given in \cref{eq:regularity_fine}--- and outputs $\tE [d(Y)]$.
By \cref{lem:sketch-fine-feasibility}, such a pseudo-expectation $\tE$ exists with high probability. 
It follows from the \sos identifiability proof in \cref{lem:sketch-fine-identifiability} that $\abs{\tE [d(Y)]/\dnull-1} \leq \tilde{O}\paren{{1}/{\sqrt{n\dnull}}+{\eta}/{\sqrt{\dnull}}}$.
Moreover, the algorithm can be implemented by semidefinite programming and run in polynomial time.

\paragraph{Private and robust algorithm via sum-of-squares exponential mechanism}
We present our private and robust algorithm in \cref{alg:sketch_fine_er_exp_mec} and give a proof sketch of \cref{thm:sketch-fine-estimation}.
\begin{algorithmbox}[Private fine estimation for \ER random graphs]
    \label{alg:sketch_fine_er_exp_mec}
    \mbox{}\\
    \textbf{Input:}
        $\eta$ corrupted random graph $A$, $\e$-differentially private coarse estimate $\hat{d}$.
    
    \noindent
    \textbf{Privacy parameter:} 
        $\eps$.
    
    \noindent
    \textbf{Output:} 
        A sample from the distribution $\mu_{A,\eps}$ with support $[0,n]$ and density 
        \begin{equation}
            \label{eq:sketch-fine_er_exp_mec}
            \mathrm{d}\mu_{A,\eps}(d) \propto \exp\paren{-\eps\cdot \soscore(d;A)} \,,
        \end{equation}
        where $\soscore(d;A)$ is defined as
\begin{equation}
    \label{eq:sketch_fine_er_score}
    \begin{split}
        \soscore(d; A) \seteq 
        \min_{0\le\gamma\le1} \gamma n \text{ s.t. }
        & \exists \text{ level-}8 \text{ pseudo-expectation } \tE \text{ satisfying } \\[-8pt]
        & \cT(Y,z; A,\gamma) \,\cup\, \cR(Y) \,\cup\, \Set{\Abs{d(Y)-d} \le 1/\poly(n)}  \,,
    \end{split}
\end{equation}
with $\cR(Y)$ given in \cref{eq:regularity_fine}.
\end{algorithmbox}

\begin{proof}[Proof sketch of \cref{thm:sketch-fine-estimation}]
    \emph{Privacy.}  
    By \cref{fact:sketch-privacy}, \cref{alg:sketch_fine_er_exp_mec} is $2\e$-differentially private. 
    \emph{Utility.} 
    For simplicity, we consider the case when there is no corruption (i.e. $\eta=0$).
    The analysis for the case when $\eta>0$ is similar.
    Let $\Anull \sim \bbG(n, \dnull/n)$.
    Then with high probability $\soscore(\dnull;\Anull)=0$.
    By a standard volume argument, \cref{alg:sketch_fine_er_exp_mec} outputs a scalar $d$ that satisfies $\soscore(d; \Anull) \lesssim \log(n) / \eps$ with high probability.
    By the definition of $\soscore$ in \cref{eq:sketch_fine_er_score}, this implies that with high probability there exists a level-$8$ pseudo-distribution satisfying $\cT(Y,z; A,\gamma) \,\cup\, \cR(Y)$ with $\gamma \lesssim \log(n) / (\e n)$.
    Taking $\eta=\log(n) / (\eps n)$ in \cref{lem:sketch-fine-identifiability}, it follows that \cref{alg:sketch_fine_er_exp_mec} outputs an estimate $\tilde{d}$ such that
    \( \abs{\tilde{d}/{\dnull}-1}\leq \tilde{O}\paren{{1}/{\sqrt{n\dnull}}+ 1/(\e n\sqrt{\dnull}}) \) with high probability.
\end{proof}

%% file: proof-sketch-lower-bound.tex
\subsection{Lower bound}
\label{sec:proof-sketch-lower-bound}

We sketch the proof of \cref{thm:lower_bound_ER}. 
Let $\alpha\in[0,1]$ and $d=(1-\alpha)\dnull$.
We can construct a coupling $\omega$ between the distributions $\bbG(n, d/n)$ and $\bbG(n,{\dnull}/{n})$ with the following property.
For $(\rv G, \rv G') \sim \omega$, we have $\dist(\rv G, \rv G')$ bounded by $\tilde{O}\paren{\alpha n\sqrt{\dnull}}$ with overwhelmingly high probability.
By the definition of differential privacy, when $\e \alpha n\sqrt{\dnull}\leq 1/\polylog(n)$, the output of an $\eps$-differentially private algorithm are indistinguishable under $\bbG(n,{d}/{n})$ and $\bbG(n,{\dnull}/{n})$. 
Therefore, by setting $\alpha=\tilde{O}(1/\e n\sqrt{\dnull})$, we conclude that no $\eps$-differentially private algorithm can achieve error rate better than $\tilde{O}(1/\e n\sqrt{\dnull})$.
This provides a matching lower bound for our private edge density estimation algorithm.

%% file: sos-exp-mechanism.tex
\section{Sum-of-squares exponential mechanism}
\label{sec:sos_exp_mechanism}

In this section, we present our sum-of-squares exponential mechanism and prove its properties in a general setting that incorporates all special cases in \cref{sec:coarse_estimation}, \cref{sec:fine_estimation_inhmomo} and \cref{sec:fine_estimation_ER}.

\paragraph{Setup}
Let $\cD\subset\R^N$.
Given an $n$-by-$n$ symmetric matrix $A$, our goal is to output an element $d$ from $\cD$ privately.
We say two symmetric matrices are neighboring if they differ in at most one row and one column.
The utility of an element $d\in\cD$ is quantified by a score function defined as follows.

\paragraph{Score function}
For an $n$-by-$n$ symmetric matrix $A$ and a scalar $\gamma$, consider the following polynomial system with indeterminates $(Y_{ij})_{i,j\in[n]}$, $(z_i)_{i\in[n]}$ and coefficients that depend on $A, \gamma$:
\begin{equation}
    \cQ_1(Y,z; A,\gamma) \seteq 
    \left. 
        \begin{cases}
            z \odot z = z,\, \iprod{\one, z} \ge (1-\gamma)n \\
            0 \leq Y \leq \one \one^{\top},\, Y=Y^\top \\
            Y \odot zz^\top = A \odot zz^\top
        \end{cases}
    \right\} \,.
\end{equation}
For an element $d\in\cD$, let $\cQ_2(Y; d)$ be a polynomial system with coefficients depending on $d$ (and independent of $A,\gamma$).
Then for a matrix $A$ and an element $d\in\cD$, we define the score of $d$ with regard to $A$ to be
\begin{align}
    \label{eq:sos_score}
    s(d; A) \seteq 
    \min_{0\leq\gamma\leq1} \gamma n \text{ s.t. } 
    \begin{cases} 
        \exists \text{ level-}\ell \text{ pseudo-expectation } \tE \text{ satisfying } \\
        \cQ_1(Y,z; A,\gamma) \,\cup\,  \cQ_2(Y; d) \,.
    \end{cases} 
\end{align}
For $s(d;A)$ to be well-defined, we assume that for every $d\in\cD$ there exists a symmetric matrix $A^* \in [0,1]^{n\times n}$ such that $\cQ_2(A^*;d)$ is true.

\begin{remark}[Score function computation]
    \label{remark:score_function_computation}
    Observe that a level-$\ell$ pseudo-expectation satisfying $\cQ_1(Y,z; A,\gamma) \,\cup\,  \cQ_2(Y; d)$ is also a level-$\ell$ pseudo-expectation satisfying $\cQ_1(Y,z; A,\gamma') \,\cup\,  \cQ_2(Y; d)$ for any $\gamma'\ge\gamma$.
    Thus we can compute $s(d;A)$ using binary search.
    Given a scalar $\gamma$, checking if there exists a level-$\ell$ pseudo-expectation satisfying $\cQ_1(Y,z; A,\gamma) \,\cup\,  \cQ_2(Y; d)$ is equivalent to checking if a semidefinite program of size $n^{O(\ell)}$ is feasible.
    Since we only have efficient algorithms for semidefinite programming up to a given precision, we can only efficiently search for pseudo-distributions that \emph{approximately} satisfy a given polynomial system.
    In spite of this, as long as the bit complexity of the coefficients in our sum-of-squares proof are polynomially bounded, the analysis of our algorithm based on sum-of-squares proofs will still work due to our discussion in \cref{remark:sos-numerical-issue}.
    We refer interested readers to \cite{Hopkins2023Privacy} for a formal (and quite technical) treatment of approximate pseudo-expectations.
\end{remark}

\paragraph{Exponential mechanism}
Given an $n$-by-$n$ symmetric matrix $A$, our sos exponential mechanism with privacy parameter $\eps$ outputs a sample from the distribution $\mu_{A,\eps}$ that is supported on $\cD$ and has density 
\begin{equation}
    \label{eq:sos_exp_mechanism}
    \mathrm{d}\mu_{A,\eps}(d) \propto \exp\paren{-\eps\cdot s(d;A)} \,.
\end{equation}

\begin{remark}[Sampling]
    \label{remark:sampling}
    To efficiently sample from $\mu_{A,\eps}$, we can use the following straightforward discretization scheme.
    More specifically, given a discretization parameter $\delta$, we output an element $d\in\set{0, \delta, 2\delta, \dots, \floor{n/\delta}\delta}$ with probability proportional to $\exp\paren{-\eps\cdot \soscore(d;A)}$.
    As the error introduced by discretization is at most $\delta$ and our target estimation error is $\omega(1/n)$, we can choose $\delta=1/n$ and the discretization error is then negligible.
    Moreover, our algorithm requires at most $n^2$ evaluations of score functions.
\end{remark}

\paragraph{Properties}
The following lemma shows that the sensitivity of score function $s(d;A)$ is at most $1$.

\begin{lemma}[Sensitivity bound]
    \label{lem:sos_score_sensitivity}
    For any $d\in\cD$ and any two $n$-by-$n$ symmetric matrices $A,A'$ that differ in at most one row and one column,
    the score function defined in \cref{eq:sos_score} satisfies
    \[
        \Abs{s(d;A) - s(d;A')} \le 1 \,.
    \]
\end{lemma}
\begin{proof}
    Without loss of generality, we assume that $A$ and $A'$ differ in the first row and column.
    Consider the linear functions $(\ell_i)$ where $\ell_1(z)=0$ and $\ell_i(z)=z_i$ for $i\ge2$.
    Then for every polynomial inequality $q(Y,z)\geq0$ in $\cQ_1(Y,z; A',\gamma+1/n) \,\cup\, \cQ_2(Y; d)$,
    \[
        \cQ_1(Y,z; A,\gamma) \,\cup\, \cQ_2(Y; d)
        \; \sststile{\deg(q)}{Y,z} \;
        q(Y,\ell(z)) \ge 0 \,.
    \]
    The same argument also holds for polynomial equalities.
    Then by \cite[Lemma 8.1]{chen2024private}, $s(d;A') \le s(d;A)+1$.
    Due to symmetry of $A$ and $A'$, we also have $s(d;A) \le s(d;A')+1$.
    Therefore, $\abs{s(d;A) - s(d;A')} \le 1$.
\end{proof}

The following privacy guarantee of our sos exponential mechanism is a direct corollary of \cref{lem:sos_score_sensitivity}.

\begin{lemma}[Privacy]
    \label{lem:sos_exp_mec_privacy}
    The exponential mechanism defined in \cref{eq:sos_exp_mechanism} is $2\eps$-differentially node private.
\end{lemma}

\begin{lemma}[Volume of low-score points]
    \label{lem:sos_exp_mec_utility}
    Let $A\in\R^{n\times n}$ and $\eps>0$.
    Consider the distribution $\mu_{A,\eps}$ defined by \cref{eq:sos_exp_mechanism}.
    Suppose $(Y=A^*, z=z^*)$ is a solution to $\cQ_1(Y,z; A,\gamma^*)$.
    Then for any $t\geq0$,
    \[
        \Pr_{\rv d\sim\mu_{A,\eps}}\Paren{s(\rv d;A) \ge \gamma^* n + \frac{t\log n}{\eps}} 
        \le \frac{\vol(\cD)}{\vol\Paren{\cG(A^*)}} \cdot n^{-t} \,,
    \]
    where $\cG(A^*) \seteq \set{d\in\cD \,:\, \cQ_2(A^*; d) \text{ is true}}$.
\end{lemma}
\begin{proof}
    Note $(Y=A^*, z=z^*)$ is also a solution to $\cQ_1(Y,z; A,\gamma^*) \,\cup\, \cQ_2(Y; d)$ for any $d$ such that $\cQ_2(A^*; d)$ is true.
    Let $\cG(A^*) \seteq \set{d\in\cD \,:\, \cQ_2(A^*; d) \text{ is true}}$.
    Thus $s(d;A)\le\gamma^* n$ for any $d \in \cG(A^*)$.
    For $t\geq0$,
    \[
        \Pr_{\rv d\sim\mu_{A,\eps}}\Paren{s(\rv d;A) \ge \gamma^* n + \frac{t\log n}{\eps}} 
        \le \frac{\vol(\cD) \cdot \exp\Paren{-\eps\gamma^* n - t\log n}}{\vol\Paren{\cG(A^*)} \cdot \exp\Paren{-\eps\gamma^* n}}
        = \frac{\vol(\cD)}{\vol\Paren{\cG(A^*)}} \cdot n^{-t} \,.
    \]
\end{proof}

%% file: coarse-estimation.tex
\section{Coarse estimation}
\label{sec:coarse_estimation}

In this section, we describe our coarse estimation algorithm that achieves constant multiplicative approximation of the expected average degree $\dnull$.

\begin{theorem}[Coarse estimation for inhomogeneous random graphs]
    \label{thm:coarse_estimation_inhomo}
    Let $\Qnull$ be an $n$-by-$n$ edge connection probability matrix and let $\dnull \seteq d(\Qnull)$. 
    Suppose $\normi{\Qnull} \le R\dnull/n$ for some $R$.
    There are constants $C_1, C_2, C_3$ such that the following holds.
    For any $\eta,\eps,\dnull$ such that $\eta\log(1/\eta)R \le C_1$, $\eps \ge C_2\log^2(n)R/n$, and $\dnull\geq C_3$, there exists a polynomial-time $\eps$-differentially node private algorithm which, given an $\eta$-corrupted inhomogeneous random graph $\bbG(n,\Qnull)$, outputs an estimate $\hat{d}$ satisfying $\abs{\hat{d}/\dnull-1} \le 0.5$ with probability $1-n^{-\Omega(1)}$.
\end{theorem}

We make a few remarks on \cref{thm:coarse_estimation_inhomo}.
\begin{itemize}
    \item Our algorithm in \cref{thm:coarse_estimation_inhomo} is a sum-of-squares exponential mechanism. $R,\eta,\eps$ are parameters given as inputs to our algorithm.
    
    \item We can get a constant estimate of $p^{\circ}$ by taking $\hat{p} = \frac{\hat{d}}{n-1}$. Since $\frac{\hat{p}}{\pnull} = \frac{\hat{d}}{\dnull}$, it follows that $\abs{\frac{\hat{p}}{\pnull}-1} \le 0.5$.

    \item When $\Qnull = \pnull (\one\one^\top - \Id)$, the inhomogeneous random graph $\bbG(n,\Qnull)$ is just the \ER random graph $\bbG(n,\pnull)$. Thus, by setting $R=\frac{n}{n-1}$ in \cref{thm:coarse_estimation_inhomo}, we directly obtain a coarse estimation result for \ER random graphs.

    \item The utility guarantee of our algorithm holds in the constant-degree regime (i.e. $\dnull\ge\Omega(1)$).
    To the best of our knowledge, even without privacy requirement and in the special case of \ER random graphs, no previous algorithm can match our guarantees in the constant-degree regime.
    Specifically, when $\dnull \ll \log n$ and $\eta\ge\Omega(1)$, the robust algorithm in \cite{acharya2022robust} can not provide a constant-factor approximation of $\dnull$.
\end{itemize}

In \cref{sec:coarse_sos}, we set up polynomial systems that our algorithm uses and prove useful sos inequalities.
In \cref{sec:robust_coarse_estimation}, we show that we can easily obtain a robust algorithm via sos proofs in \cref{sec:coarse_sos}.
Then in \cref{sec:private_coarse_estimation}, we describe our algorithm and prove \cref{thm:coarse_estimation_inhomo}.

\input{coarse-estimation-sos}

\input{coarse-estimation-robust-algorithm}

\input{coarse-estimation-private-algorithm}

%% file: coarse-estimation-sos.tex
\subsection{Sum-of-squares}
\label{sec:coarse_sos}

For an adjacency matrix $A$ and two nonnegative scalars $\gamma$ and $\sigma$, consider the following polynomial systems with indeterminates $Y=(Y_{ij})_{i,j\in[n]}$, $z=(z_i)_{i\in[n]}$ and coefficients that depend on $A,\gamma,\sigma$:
\begin{gather}
  \cP_1(Y,z; A,\gamma) \seteq 
  \left. 
    \begin{cases}
      z \odot z = z,\, \iprod{\one, z} \ge (1-\gamma)n \\
      0 \leq Y \leq \one \one^{\top},\, Y=Y^\top \\
      Y \odot zz^\top = A \odot zz^\top
    \end{cases} 
  \right\} \,, \label{poly_sys:coarse_subgraph_intersect}
  \\
  \cP_2(Y; \sigma) \seteq
  \left. 
    \begin{cases}
      d(Y) = \Iprod{Y,\one\one^\top} / n \\
      (Y\one)_i \le \sigma d(Y) & \forall i\in[n]
    \end{cases}
  \right\} \,. \label{poly_sys:coarse_max_degree_bound}
\end{gather}

For convenience of notation, we will consider the following combined polynomial system in remaining of the section
\begin{equation}
    \label{poly_sys:coarse_subgraph_intersect_and_max_degree_bound}
    \cC(Y,z;A,\gamma,\sigma) \seteq \cP_1(Y,z; A,\gamma) \,\cup\, \cP_2(Y;\sigma) \,.
\end{equation}

\begin{lemma}
\label{lem:coarse-sos}
    If $(A^*, z^*)$ is a feasible solution to $\cC(Y,z;A,\gamma^*,\sigma)$ and $1-2 \gamma \sigma-2 \gamma^* \sigma > 0$, then it follows that
    \begin{equation*}
        \cC(Y,z;A,\gamma,\sigma)
        \sststile{8}{Y, z} (1-2 \gamma \sigma-2 \gamma^* \sigma) d(A^*) \leq d(Y) \leq \frac{1}{1-2 \gamma \sigma-2 \gamma^* \sigma} d(A^*)\,.
    \end{equation*}
\end{lemma}

\begin{proof}
Let $w = z \odot z^*$, by constraint $Y \odot zz^\top = A \odot zz^\top$ and $A^* \odot z^*(z^*)^\top = A \odot z^*(z^*)^\top$, we have
\begin{align*}
    \cC
    \sststile{4}{Y, z} Y \odot w w^{\top}
    & = Y \odot zz^\top \odot z^*(z^*)^\top \\
    & = A \odot zz^\top \odot z^*(z^*)^\top \\
    & = A \odot z^*(z^*)^\top \odot zz^\top \\
    & = A^* \odot z^*(z^*)^\top \odot zz^\top \\
    & = A^* \odot w w^{\top} \,.
\end{align*}
Applying this equality, it follows that
\begin{align*}
    \cC
    \sststile{4}{Y, z} n \cdot d(Y)
    & = \Iprod{Y,\one\one^\top} \\
    & = \Iprod{Y,w w^\top} + \Iprod{Y,\one\one^\top - w w^\top} \\
    & = \Iprod{A^*,w w^\top} + \Iprod{Y,2(\one-w)\one^{\top}} - \Iprod{Y,(\one-w)(\one-w)^{\top}} \,.
\end{align*}
For the first term, since $A^*_{i, j} \in [0, 1]$, $z^*_{i} \in  \set{0, 1}$ and $\cC \sststile{2}{Y, z} 0 \leq z_i \leq 1$ for all $i, j \in [n]$, we have
\begin{equation*}
    \cC
    \sststile{4}{Y, z}
    A^*_{i, j} w_{i} w_{j} = A^*_{i, j} z^*_{i} z^*_{j} z_{i} z_{j} \leq A^*_{i, j} \,.
\end{equation*}
Therefore, it follows that
\begin{equation}
    \label{eq:coarse-SOS3}
    \cC
    \sststile{4}{Y, z}
    \Iprod{A^*,w w^\top} \leq \Iprod{A^*,\one \one^\top} \leq n \cdot d(A^*) \,.
\end{equation}
For the second term, we have
\begin{align*}
    \cC
    \sststile{4}{Y, z}
    \Iprod{Y,2(\one-w)\one^{\top}}
    & = \Iprod{Y\one,2(\one-w)} \\
    & = \sum_{i \in [n]} 2 (1-w_i) \cdot (Y\one)_i \\
    & = \sum_{i \in [n]} 2 (1-z_i z^*_i) \cdot (Y\one)_i \\
    & \leq \sum_{i \in [n]} 2 (1-z_i) \cdot (Y\one)_i + \sum_{i \in [n]} 2 (1-z^*_i) \cdot (Y\one)_i \,,
\end{align*}
where the last inequality is due to \cref{lem:preliminary-SOS-set-union-bound}. From constraints $\sum_{i \in [n]} 1-z^*_i \leq \gamma^* n$, $\sum_{i \in [n]} 1-z_i \leq \gamma n$ and $(Y\one)_i \le \sigma d(Y)$ for all $i \in [n]$, it follows that
\begin{align}
\label{eq:coarse-SOS4}
\begin{split}
    \cC
    \sststile{4}{Y, z}
    \Iprod{Y,2(\one-w)\one^{\top}}
    & \leq \sum_{i \in [n]} 2 (1-z_i) \cdot \sigma d(Y) + \sum_{i \in [n]} 2 (1-z^*_i) \cdot \sigma d(Y) \\
    & = 2 \sigma d(Y) \cdot \Paren{\sum_{i \in [n]} 1-z_i} + 2 \sigma d(Y) \cdot \Paren{\sum_{i \in [n]} 1-z^*_i} \\
    & \leq 2 \gamma n \sigma d(Y) + 2 \gamma^* n \sigma d(Y) \,.
\end{split}
\end{align}
For the third term, since $\cC \sststile{2}{Y, z} Y_{i,j} \geq 0$ and $\cC \sststile{2}{Y, z} 1-w_i \geq 0$ for all $i, j \in [n]$, it follows that
\begin{equation}
    \label{eq:coarse-SOS5}
    \cC
    \sststile{8}{Y, z}
    \Iprod{Y,(\one-w)(\one-w)^{\top}} \geq 0 \,.
\end{equation}
Combining \cref{eq:coarse-SOS3}, \cref{eq:coarse-SOS4} and \cref{eq:coarse-SOS5}, we can get
\begin{align*}
    \cC
    & \sststile{8}{Y, z} n \cdot d(Y)
    \leq n \cdot d(A^*) + 2 \gamma n \sigma d(Y) + 2 \gamma^* n \sigma d(Y) \\
    & \sststile{8}{Y, z} d(Y)
    \leq \frac{d(A^*)}{1-2\gamma \sigma-2\gamma^* \sigma}\,.
\end{align*}
Swapping the roll of $A^*$ and $Y$, we can use the same proof to get
\begin{align*}
    \cC
    & \sststile{8}{Y, z} n \cdot d(A^*)
    \leq n \cdot d(Y) + 2 \gamma n \sigma d(A^*) +2 \gamma^* n \sigma d(A^*) \\
    & \sststile{8}{Y, z}
    (1-2\gamma \sigma-2\gamma^* \sigma) d(A^*) \leq d(Y)\,.
\end{align*}
This completes the proof.
\end{proof}

\begin{lemma}
\label{lem:coarse-existence-proof}
  Let $Q^{\circ}$ be an $n$-by-$n$ edge connection probability matrix and $d^{\circ} \seteq d(Q^{\circ})$. Suppose $\Norm{Q^{\circ}}_{\infty} \le R d^{\circ} / n$ for $R \in \R$.
  Let $A$ be an $\eta$-corrupted adjacency matrix of a random graph $\rv G^{\circ} \sim \bbG(n,Q^{\circ})$.
  Suppose $\eta\log(1/\eta)R \le C_1$ for some constant $C_1$ that is small enough.
  With probability $1-n^{-\Omega(1)}$, there exists $A^*$ and $z^*$ such that
    \begin{enumerate}
        \item $|d(A^*) - d^{\circ}| \leq 0.1 d^{\circ}$.
        \item $(A^*, z^*)$ is a feasible solution to $\cC(Y,z;A,\gamma, \sigma)$ with $\gamma = 2\eta$ and $\sigma = 2 \log(1/\eta)R$.
    \end{enumerate}
\end{lemma}

\begin{proof}
  Let $\rv A^{\circ}$ be the adjacency matrix of $\rv G^{\circ}$ and $z^{\circ} \in \set{0, 1}^n$ denote the set of uncorrupted nodes ($z^{\circ}_i = 1$ if and only if node $i$ is uncorrupted).
  
  By \cref{lem:degree_distribution} and \cref{lem:degree_pruning}, we know that, with probability $1-n^{-\Omega(1)}$, there exists a degree-pruned adjacency matrix $\Tilde{A}$ such that
  \begin{enumerate}
      \item $\Norm{\Tilde{A} \one}_{\infty} \leq \log(1/\eta) R d^{\circ}$.
      \item At most $\eta n$ nodes are pruned.
      \item At most $2 \eta \log(1/\eta) n R d^{\circ}$ edges are pruned.
  \end{enumerate}
  
  Let $\Tilde{z} \in \set{0, 1}^n$ denote the set of unpruned nodes ($z^{\circ}_i = 1$ if and only if node $i$ is not pruned). We will show that $A^* = \Tilde{A}$ and $z^* = z^{\circ} \odot \Tilde{z}$ satisfies the lemma.

  \paragraph{Guarantee 1}
  By \cref{lem:average_degree_concentration}, we know that, with probability $1-n^{-\Omega(1)}$,
  \begin{equation}
  \label{eq:coarse-robustness-1}
      |d(\rv A^{\circ}) - d^{\circ}| \leq 10 \sqrt{\frac{d^{\circ} \log n}{n}} \,.
  \end{equation}
  From degree pruning guarantee (3), we have that
  \begin{equation}
  \label{eq:coarse-robustness-2}
      |d(\Tilde{A}) - d(\rv A^{\circ})| \leq 2 \eta \log(1/\eta) n R d^{\circ} \,.
  \end{equation}
  Combining \cref{eq:coarse-robustness-1} and \cref{eq:coarse-robustness-2}, for some constant $C_1$ that is small enough, we have
  \begin{align}
  \label{eq:coarse-robustness-3}
  \begin{split}
      |d(\Tilde{A}) - d^{\circ}|
      & \leq |d(\Tilde{A}) - d(\rv A^{\circ})| + |d(\rv A^{\circ}) - d^{\circ}| \\
      & \leq 10 \sqrt{\frac{d^{\circ} \log n}{n}} + 2 \eta \log{\frac{1}{\eta}} R d^{\circ} \\
      & \leq 10 \sqrt{\frac{\log n}{n}} d^{\circ} + 2 C_1 d^{\circ}\\
      & \leq 0.1 d^{\circ} \,.
  \end{split}
  \end{align}

  \paragraph{Guarantee 2}
  It is easy to check that $z^* \odot z^* = z^*$, $0 \leq A^* \leq \one \one^{\top}$ and $A^* = (A^*)^\top$. Since $\iprod{\one, \Tilde{z}} \geq 1 - \eta n$ by degree pruning condition (2) and $\iprod{\one, z^{\circ}} \geq 1 - \eta n$ by corruption rate, it is easy to verify that
  \begin{equation*}
      \iprod{\one, z^*} \geq 1 - 2 \eta n \,.
  \end{equation*}
  Moreover, we have $A^* \odot z^* (z^*)^{\top} = A \odot z^* (z^*)^{\top}$ due to
  \begin{equation*}
      \Tilde{A} \odot \Tilde{z} \Tilde{z}^{\top} \odot z^{\circ} (z^{\circ})^{\top}
      = \rv A^{\circ} \odot \Tilde{z} \Tilde{z}^{\top} \odot z^{\circ} (z^{\circ})^{\top}
      = \rv A^{\circ} \odot z^{\circ} (z^{\circ})^{\top} \odot \Tilde{z} \Tilde{z}^{\top}
      = A \odot z^{\circ} (z^{\circ})^{\top} \odot \Tilde{z} \Tilde{z}^{\top} \,.
  \end{equation*}
  From \cref{eq:coarse-robustness-3}, we can get that $d^{\circ} \leq 2 d(\Tilde{A})$. Plugging this into degree pruning condition (1), we get
  \begin{equation*}
      \Norm{\Tilde{A} \one}_{\infty}
      \leq \log(1/\eta) R d^{\circ}
      \leq 2 \log(1/\eta) R d(\Tilde{A}) \,.
  \end{equation*}
  Therefore, we have
  \begin{equation*}
      (A^* \one)_i \leq 2 \log(1/\eta) R d(A^*) \,.
  \end{equation*}
  for all $i \in [n]$.
  
  Thus, $(A^*, z^*)$ is a feasible solution to $\cC(Y,z;A,\gamma, \sigma)$ with $\gamma = 2\eta$ and $\sigma = 2 \log(1/\eta) R$.
\end{proof}

%% file: coarse-estimation-robust-algorithm.tex
\subsection{Robust algorithm}
\label{sec:robust_coarse_estimation}

In this section, we show that the following algorithm based on sum-of-squares proofs in \cref{sec:coarse_sos} obtains a robust constant multiplicative approximation of $\dnull$.

\begin{algorithmbox}[Robust coarse estimation algorithm]
    \label{alg:robust-coarse}
    \mbox{}\\
    \textbf{Input:}
        $\eta$-corrupted adjacency matrix $A$, corruption fraction $\eta$ and parameter $R$.

    \noindent
    \textbf{Algorithm: } Obtain level-$8$ pseudo-expectation $\tilde{\E}$ by solving sum-of-squares relaxation of program $\cC(Y,z;A,\gamma,\sigma)$ (defined in \cref{poly_sys:coarse_subgraph_intersect_and_max_degree_bound}) with $A$, $\gamma = 2\eta$ and $\sigma = 2 \log(1/\eta) R$.
    
    \noindent
    \textbf{Output:} $\tilde{\E}[d(Y)]$
\end{algorithmbox}

\begin{theorem}[Robust coarse estimation]
    \label{thm:robust-coarse-estimation}
    Let $\Qnull$ be an $n$-by-$n$ edge connection probability matrix and let $\dnull\seteq d(\Qnull)$. 
    Suppose $\normi{\Qnull} \le R\dnull/n$ for some $R$.
    Let $A$ be an $\eta$-corrupted adjacency matrix of a random graph $\rv G^{\circ} \sim \bbG(n,Q^{\circ})$.
    Suppose $\eta\log(1/\eta)R \le c$ for some constant $c$ that is small enough.
    With probability $1-n^{-\Omega(1)}$, \cref{alg:robust-coarse} outputs an estimate $\hat{d}$ satisfying $\abs{\frac{\hat{d}}{\dnull}-1} \le 0.5$.
\end{theorem}

\begin{proof}
  By \cref{lem:coarse-sos} and \cref{lem:coarse-existence-proof}, we know that
  \begin{equation*}
    \cC(Y,z;A,\gamma,\sigma)
    \sststile{8}{Y, z} (1-4 \gamma \sigma) d(A^*) \leq d(Y) \leq \frac{1}{1-4 \gamma \sigma} d(A^*) \,,
  \end{equation*}
  and,
  \begin{equation*}
    |d(A^*) - d^{\circ}| \leq 0.1 d^{\circ} \,.
  \end{equation*}
  Therefore, we have 
    \begin{equation*}
    \cC(Y,z;A,\gamma,\sigma)
    \sststile{8}{Y, z} 0.9 (1-4 \gamma \sigma)d^{\circ} \leq d(Y) \leq \frac{1.1}{1-4 \gamma \sigma} d^{\circ} \,.
  \end{equation*}
  Consider $4 \gamma \sigma$, for constant $c$ that is small enough, we have
  \begin{equation*}
      4 \gamma \sigma = 8 \eta \log(1/\eta) R \leq 8 c \leq 0.1\,.
  \end{equation*}
  This implies that $0.9 (1-4 \gamma \sigma) \geq \frac{1}{2}$ and $\frac{1.1}{1-4 \gamma \sigma} \leq \frac{11}{9} \leq \frac{3}{2}$. Therefore, we have
  \begin{equation*}
    \cC(Y,z;A,\gamma,\sigma)
    \sststile{8}{Y, z} \frac{1}{2} d^{\circ} \leq d(Y) \leq \frac{3}{2} d^{\circ} \,.
  \end{equation*}
  Thus, the level-$8$ pseudo-expectation $\tilde{\E}$ satisfies
  \begin{equation*}
      \frac{1}{2} d^{\circ} \leq \tilde{\E}[d(Y)] \leq \frac{3}{2} d^{\circ} \,,
  \end{equation*}
  which implies that
  \begin{equation*}
      \Abs{\frac{\tilde{\E}[d(Y)]}{d^{\circ}} - 1} \leq \frac{1}{2} \,.
  \end{equation*}
\end{proof}

%% file: coarse-estimation-private-algorithm.tex
\subsection{Private algorithm}
\label{sec:private_coarse_estimation}

In this section, we present our algorithm and prove \cref{thm:coarse_estimation_inhomo}. 
Our algorithm instantiates the sum-of-squares exponential mechanism in \cref{sec:sos_exp_mechanism}.

\paragraph{Score function}
For an $n$-by-$n$ symmetric matrix $A$ and a scalar $d$, we define the score of $d$ with regard to $A$ to be
\begin{align}
    \label{eq:coarse_sos_score}
    s(d; A) \seteq 
    \min_{0\le\gamma\le1} \gamma n \text{ s.t. } 
    \begin{cases} 
        \exists \text{ level-}8 \text{ pseudo-expectation } \tE \text{ satisfying } \\
        \cC(Y,z;A,\gamma,\sigma) \,\cup\, \Set{\Abs{d(Y)-d}\le\alpha d} \,,
    \end{cases}
\end{align}
where $\cC(Y,z;A,\gamma,\sigma)$ is the polynomial system defined in \cref{poly_sys:coarse_subgraph_intersect_and_max_degree_bound}, and
$\sigma,\alpha$ are fixed parameters whose values will be decided later.
Note that $(Y=\frac{d}{n}\one\one^\top, z=\zero)$ is a solution to the polynomial system $\cC(Y,z;A,1,\sigma) \,\cup\, \Set{\Abs{d(Y)-d}\le\alpha d}$ for any $A\in\R^{n\times n}$, $d\in[0,n]$, and $\sigma\ge1$.

To efficiently compute $s(d;A)$, we can use the scheme as described in \cref{remark:score_function_computation}.

\paragraph{Exponential mechanism}
Given a privacy parameter $\eps$ and an $n$-by-$n$ symmetric matrix $A$, our algorithm is the exponential mechanism with score function \cref{eq:coarse_sos_score} and range $[0,n]$.

\begin{algorithmbox}[Coarse estimation] %
    \label{alg:coarse_exp_mec}
    \mbox{}\\
    \textbf{Input:}
        Graph $A$.
    
    \noindent
    \textbf{Parameters:} 
        $\eps, \sigma, \alpha$.
    
    \noindent
    \textbf{Output:} 
        A sample from the distribution $\mu_{A,\eps}$ with support $[0,n]$ and density 
        \begin{equation}
            \label{eq:coarse_exp_mec}
            \mathrm{d}\mu_{A,\eps}(d) \propto \exp\paren{-\eps\cdot s(d;A)} \,,
        \end{equation}
        where $s(d;A)$ is defined in \cref{eq:coarse_sos_score}.
\end{algorithmbox}

To efficiently sample from $\mu_{A,\eps}$, we can use the scheme as described in \cref{remark:sampling}.

\paragraph{Privacy}
The following privacy guarantee of our algorithm is a direct corollary of \cref{lem:sos_exp_mec_privacy}.

\begin{lemma}[Privacy]
    \label{lem:coarse_exp_mec_privacy}
    \cref{alg:coarse_exp_mec} is $2\eps$-differentially node private.
\end{lemma}

\paragraph{Utility}
The utility guarantee of our algorithm is stated in the following lemma.

\begin{lemma}[Utility]
    \label{lem:coarse_exp_mec_util}
    Let $\Qnull$ be an $n$-by-$n$ edge connection probability matrix and let $\dnull\seteq d(\Qnull)$. 
    Suppose $\normi{\Qnull} \le R\dnull/n$ for some $R$.
    There are constants $C_1, C_2, C_3$ such that the following holds.
    For any $\eta,\eps,\dnull$ such that $\eta\log(1/\eta)R \le C_1$, $\eps \ge C_2\log^2(n)R/n$, and $\dnull\geq C_3$, 
    given an $\eta$-corrupted inhomogeneous random graph $\bbG(n,\Qnull)$,
    \cref{alg:coarse_exp_mec} outputs an estimate $\hat{d}$ satisfying $\abs{\hat{d} - \dnull} \le 0.5\dnull$ with probability $1-n^{-\Omega(1)}$.
\end{lemma}

Before proving \cref{lem:coarse_exp_mec_util}, we need the following two lemmas.

\begin{lemma}[Volume of low-score points]
    \label{lem:coarse_exp_mec_volume_of_low_score_points}
    Let $A\in\R^{n\times n}$ and $\eps>0$.
    Consider the distribution $\mu_{A,\eps}$ defined by \cref{eq:coarse_exp_mec}.
    Suppose $(Y=A^*, z=z^*)$ is a solution to $\cC(Y,z; A,\gamma^*,\sigma)$ and $d(A^*)\ge2$.
    Then for any $t\geq0$,
    \[
        \Pr_{ \rv{d} \sim \mu_{A,\eps} } \Paren{ s(\rv{d};A) \ge \gamma^* n + \frac{t\log n}{\eps} } 
        \le \frac{ n^{-t+1} }{\alpha} \,.
    \]
\end{lemma}
\begin{proof}
    Apply \cref{lem:sos_exp_mec_utility} with $\cD=[0,n]$ and
    \[
        \cG(A^*) = \Set{d\in\cD \,:\, \frac{d(A^*)}{1+\alpha} \le d \le \frac{d(A^*)}{1-\alpha}} \,.
    \]
    As $[d(A^*)/(1+\alpha), d(A^*)] \subseteq \cG(A^*)$ and $d(A^*) \ge 2 \ge 1+\alpha$, we have $\vol(\cG(A^*)) \ge \alpha$.
\end{proof}

\begin{lemma}[Low score implies utility]
    \label{lem:coarse_exp_mec_low_score_implies_utility}
    Let $A\in\R^{n\times n}$ and consider the score function $s(\cdot \,; A)$ defined in \cref{eq:coarse_sos_score}.
    Suppose $(Y=A^*, z=z^*)$ is a solution to $\cC(Y,z; A, \gamma^*, \sigma)$.
    For a scalar $d$ such that $s(d;A) \le \tau n$ and $(\gamma^* + \tau)\sigma \le 0.1$,
    \[
        \frac{0.8}{1+\alpha} d(A^*) 
        \leq d \leq 
        \frac{1.25}{1-\alpha} d(A^*) \,.
    \]
\end{lemma}
\begin{proof}
    Applying \cref{lem:coarse-sos} with $(\gamma^* + \tau)\sigma \le 0.1$, we have
    \[
        \cC(Y,z; A,\tau,\sigma)
        \; \sststile{8}{Y, z} \;
        0.8 d(A^*) \le d(Y) \le 1.25 d(A^*)
    \]
    Thus,
    \begin{align*}
        \cC(Y,z;A,\tau,\sigma) \cup \Set{\Abs{d(Y)-d} \le \alpha d}
        \; \sststile{8}{Y, z} \; 
        \frac{0.8}{1+\alpha} d(A^*) 
        \leq d \leq 
        \frac{1.25}{1-\alpha} d(A^*) \,.
    \end{align*}
\end{proof}

Now we are ready to prove \cref{lem:coarse_exp_mec_util}.

\begin{proof}[Proof of \cref{lem:coarse_exp_mec_util}]
    Let $A$ be a realization of $\eta$-corrupted $\bbG(n,\Qnull)$.
    By \cref{lem:coarse-existence-proof}, the following event happens with probability $1-n^{-\Omega(1)}$.
    There exists a solution $(Y=A^*, z=z^*)$ to $\cC(Y,z; A,\gamma^*,\sigma)$ with $\gamma^*=2\eta$, $\sigma = 2\log(1/\eta)R$, and $0.9\dnull \le d(A^*) \le 1.1\dnull$.
    
    As $d(A^*) \ge 0.9\dnull \ge 2$, then it follows by setting $t=10$ and $\alpha=0.01$ in \cref{lem:coarse_exp_mec_volume_of_low_score_points} that,
    \[
        \Pr_{\rv{d} \sim \mu_{A,\eps}} \Paren{s(\rv{d};A) \le \tau n} \ge 1 - n^{-9} 
        \text{ where } \tau \seteq 2\eta + 10\log(n)/(\eps n) \,.
    \]
    
    As an $\eta$-corrupted graph is actually uncorrupted when $\eta<1/n$, we can assume $\eta\ge1/(2n)$ without loss of generality.
    Thus,
    \[
        (2\eta+\tau)\sigma \le 8\eta\log(1/\eta)R + \frac{20\log^2(n)R}{n\eps}  \,.
    \]
    For $\eta\log(1/\eta)R$ and $\log^2(n) R / (\eps n)$ smaller than some constant, we have $(2\eta+\tau) \sigma \le 0.1$.
    Let $\hat d$ be a scalar such that $s(\hat{d}; A) \le \tau n$.
    Then by \cref{lem:coarse_exp_mec_low_score_implies_utility},
    \[
        \frac{0.8}{1+\alpha} d(A^*) \leq \hat{d} \leq \frac{1.25}{1-\alpha} d(A^*) \,.
    \]
    Plugging in $\alpha\le0.01$ and $0.9\dnull \le d(A^*) \le 1.1\dnull$, we have
    \[
        0.5\dnull \le \hat{d} \le 1.5\dnull \,.
    \]
\end{proof}

\paragraph{Proof of \cref{thm:coarse_estimation_inhomo}}
By \cref{lem:coarse_exp_mec_privacy} and \cref{lem:coarse_exp_mec_util}.

%% file: fine-estimation-inhomo.tex
\section{Fine estimation for inhomogeneous random graphs}
\label{sec:fine_estimation_inhmomo}

From \cref{sec:coarse_estimation}, we have a constant multiplicative approximation of the expected average degree $\dnull$.
In this section, we show how to use this coarse estimate to obtain our fine estimator for inhomogeneous random graphs.

\begin{theorem}[Fine estimation for inhomogeneous random graphs]
    \label{thm:fine_estimation_inhomo}
    Let $\Qnull$ be an $n$-by-$n$ edge connection probability matrix and let $\dnull \seteq d(\Qnull)$. 
    Suppose $\Norm{\Qnull}_{\infty} \le R\dnull / n$ for some $R$.
    There is a sufficiently small constant $c$ such that the following holds.
    For any $\eta$ such that $\eta\log(1/\eta)R \le c$, there exists a polynomial-time $\eps$-differentially node private algorithm which, given an $\eta$-corrupted inhomogeneous random graph $\bbG(n,\Qnull)$ and a constant-factor approximation of $\dnull$, outputs an estimate $\tilde{d}$ satisfying
    \[
        \Abs{\frac{\tilde{d}}{\dnull}-1} \le 
        O \Paren{\sqrt{\frac{\log n}{\dnull n}} +
        \frac{R \log^2 n}{\eps n} +
        \eta\log(1/\eta)R} \,,
    \]
    with probability $1-n^{-\Omega(1)}$.
\end{theorem}

We make a few remarks on \cref{thm:fine_estimation_inhomo}.
\begin{itemize}
    \item Our algorithm in \cref{thm:fine_estimation_inhomo} is a sum-of-squares exponential mechanism. $R,\eta,\eps$ are parameters given as input to our algorithm.
    
    \item We can get an estimate of $p^{\circ}$ by taking $\tilde{p} = \frac{\tilde{d}}{n-1}$. Since $\frac{\tilde{p}}{\pnull} = \frac{\tilde{d}}{\dnull}$, it follows that
    \[
        \Abs{\frac{\tilde{p}}{\pnull}-1} \le 
        O \Paren{\sqrt{\frac{\log n}{\dnull n}} +
        \frac{R \log^2 n}{\eps n} +
        \eta\log(1/\eta)R} \,.
    \]

    \item Combining \cref{thm:coarse_estimation_inhomo} and \cref{thm:fine_estimation_inhomo} gives us an efficient, private, and robust edge density estimation algorithm for inhomogeneous random graphs whose utility guarantee is information-theoretically optimal up to a factor of  $\log n$ and $\log(1/\eta)$.
\end{itemize}

In \cref{sec:fine-inhomo-sos}, we set up polynomial systems that our algorithm uses and prove useful sos inequalities.
In \cref{sec:robust_fine_estimation_inhomo}, we show that we can easily obtain a robust algorithm via sos proofs in \cref{sec:fine-inhomo-sos}.
Then in \cref{sec:private_fine_estimation_inhomo}, we describe our algorithm and prove \cref{thm:fine_estimation_inhomo}.

\input{fine-estimation-inhomo-sos}

\input{fine-estimation-inhomo-robust-algorithm}

\input{fine-estimation-inhomo-private-algorithm}

%% file: fine-estimation-inhomo-sos.tex
\subsection{Sum-of-squares}
\label{sec:fine-inhomo-sos}

For an adjacency matrix $A$ and nonnegative scalars $\gamma$, $\sigma$ and $\hat{d}$, consider the following polynomial systems with indeterminates $Y=(Y_{ij})_{i,j\in[n]}$, $z=(z_i)_{i\in[n]}$ and coefficients that depend on $A,\gamma,\sigma,\hat{d}$:
\begin{gather}
  \cP_1(Y,z; A,\gamma) \seteq 
  \left. 
    \begin{cases}
      z \odot z = z,\, \iprod{\one, z} \ge (1-\gamma)n \\
      0 \leq Y \leq \one \one^{\top},\, Y=Y^\top \\
      Y \odot zz^\top = A \odot zz^\top
    \end{cases}
  \right\} \,, \label{poly_sys:fine_inhomo_subgraph_intersect}
  \\
  \cP_3(Y; \sigma, \hat{d}) \seteq
  \left. 
    \begin{cases}
      (Y\one)_i \le \sigma \hat{d} & \forall i\in[n]
    \end{cases}
  \right\} \,. \label{poly_sys:fine_inhomo_max_degree_bound}
\end{gather}
For convenience of notation, we will consider the following combined polynomial system in the remaining of this section
\begin{equation}
\label{poly_sys:fine_inhomo_subgraph_intersect_and_max_degree_bound}
    \cD(Y,z;A,\gamma,\sigma, \hat{d}) \seteq \cP_1(Y,z; A,\gamma) \,\cup\, \cP_3(Y; \sigma, \hat{d}) \,.
\end{equation}

\begin{lemma}
\label{lem:fine-inhomo-sos}
    If $(A^*, z^*)$ is a feasible solution to $\cD(Y,z;A,\gamma^*,\sigma, \hat{d})$, then it follows that
    \begin{equation*}
        \cD(Y,z;A,\gamma,\sigma, \hat{d})
        \sststile{8}{Y, z} |d(Y) - d(A^*)| \leq 2 (\gamma + \gamma^*) \sigma \hat{d} \,.
    \end{equation*}
\end{lemma}

\begin{proof}
Let $w = z \odot z^*$. Using similar analysis as in the proof of \cref{lem:coarse-sos}, it follows that
\begin{equation*}
    \cD
    \sststile{4}{Y, z} Y \odot w w^{\top} = A^* \odot w w^{\top} \,,
\end{equation*}
and,
\begin{align*}
    \cD
    \sststile{8}{Y, z} n \cdot d(Y)
    & = \Iprod{Y,\one\one^\top} \\
    & = \Iprod{A^*,w w^\top} + \Iprod{Y,2(\one-w)\one^{\top}} - \Iprod{Y,(\one-w)(\one-w)^{\top}} \\
    & \leq \Iprod{A^*,\one \one^\top} + \Iprod{Y \one,2(\one-w)} \\
    & \leq n \cdot d(A^*) + 2 (\gamma + \gamma^*) \sigma n \hat{d} \,.
\end{align*}
By rearranging the terms, we have
\begin{equation*}
    \cD
    \sststile{8}{Y, z}
    d(Y) - d(A^*) \leq 2 (\gamma + \gamma^*) \sigma \hat{d} \,.
\end{equation*}
Swapping the roll of $Y$ and $A^*$, we can also get
\begin{equation*}
    \cD
    \sststile{8}{Y, z}
    d(A^*) - d(Y) \leq 2 (\gamma + \gamma^*) \sigma \hat{d} \,.
\end{equation*}
This completes the proof.
\end{proof}

\begin{lemma}
\label{lem:fine-inhomo-existence-proof}
  Let $Q^{\circ}$ be an $n$-by-$n$ edge connection probability matrix and $d^{\circ} \seteq d(Q^{\circ})$. Suppose $\Norm{Q^{\circ}}_{\infty} \le R d^{\circ} / n$ for $R \in \R$.
  Let $A$ be an $\eta$-corrupted adjacency matrix of a random graph $\rv G^{\circ} \sim \bbG(n,Q^{\circ})$.
  Suppose $\eta\log(1/\eta)R \le C_1$ for some constant $C_1$ that is small enough.
  With probability $1-n^{-\Omega(1)}$, there exists $A^*$ and $z^*$ such that
    \begin{itemize}
        \item $|d(A^*) - d^{\circ}| \leq 10 \sqrt{\frac{d^{\circ} \log n}{n}} + 2 \eta \log(1/\eta) R d^{\circ}$.
        \item $(A^*, z^*)$ is a feasible solution to $\cD(Y,z;A,\gamma,\sigma, \hat{d})$ with $\eta$-corrupted $A$, $\gamma = 2\eta$, $\sigma = 10 \log(1/\eta) R$ and $\hat{d} \geq \frac{1}{2} d^{\circ}$.
    \end{itemize}
\end{lemma}

\begin{proof}
  Let $\rv A^{\circ}$ be the adjacency matrix of $\rv G^{\circ}$ and $z^{\circ} \in \set{0, 1}^n$ denote the set of uncorrupted nodes ($z^{\circ}_i = 1$ if and only if node $i$ is uncorrpted).
  
  By \cref{lem:degree_distribution} and \cref{lem:degree_pruning}, we know that, with probability $1-n^{-\Omega(1)}$, there exists a degree-pruned adjacency matrix $\Tilde{A}$ such that
  \begin{enumerate}
      \item $\Norm{\Tilde{A} \one}_{\infty} \leq \log(1/\eta) R d^{\circ}$.
      \item At most $\eta n$ nodes are pruned.
      \item At most $2 \eta \log(1/\eta) n R d^{\circ}$ edges are pruned.
  \end{enumerate}
  Let $\Tilde{z} \in \set{0, 1}^n$ denote the set of unpruned nodes ($\Tilde{z}_i = 1$ if and only if node $i$ is not pruned). We will show that $A^* = \Tilde{A}$ and $z^* = z^{\circ} \odot \Tilde{z}$ satisfies the lemma.
  
  \paragraph{Guarantee 1}
  By \cref{lem:average_degree_concentration}, we know that, with probability $1-n^{-\Omega(1)}$,
  \begin{equation}
  \label{eq:fine-inhomo-robustness-1}
      |d(\rv A^{\circ}) - d^{\circ}| \leq 10 \sqrt{\frac{d^{\circ} \log n}{n}} \,.
  \end{equation}
  From degree pruning guarantee (3), we have that
  \begin{equation}
  \label{eq:fine-inhomo-robustness-2}
      |d(\Tilde{A}) - d(\rv A^{\circ})| \leq 2 \eta \log(1/\eta) R d^{\circ} \,.
  \end{equation}
  Combining \cref{eq:fine-inhomo-robustness-1} and \cref{eq:fine-inhomo-robustness-2}, we have
  \begin{align*}
      |d(\Tilde{A}) - d^{\circ}|
      & \leq |d(\Tilde{A}) - d(A^{\circ})| + |d(A^{\circ}) - d^{\circ}| \\
      & \leq 10 \sqrt{\frac{d^{\circ} \log n}{n}} + 2 \eta \log(1/\eta) R d^{\circ} \,.
  \end{align*}
  
  \paragraph{Guarantee 2}
  It is easy to check that $z^* \odot z^* = z^*$, $0 \leq A^* \leq \one \one^{\top}$ and $A^* = (A^*)^\top$. Since $\iprod{\one, \Tilde{z}} \geq 1 - \eta n$ by degree pruning condition (2) and $\iprod{\one, z^{\circ}} \geq 1 - \eta n$ by corruption rate, it is easy to verify that
  \begin{equation*}
      \iprod{\one, z^*} \geq 1 - 2 \eta n \,.
  \end{equation*}
  Moreover, we have $A^* \odot z^* (z^*)^{\top} = A \odot z^* (z^*)^{\top}$ due to
  \begin{equation*}
      \Tilde{A} \odot \Tilde{z} \Tilde{z}^{\top} \odot z^{\circ} (z^{\circ})^{\top}
      = \rv A^{\circ} \odot \Tilde{z} \Tilde{z}^{\top} \odot z^{\circ} (z^{\circ})^{\top}
      = \rv A^{\circ} \odot z^{\circ} (z^{\circ})^{\top} \odot \Tilde{z} \Tilde{z}^{\top}
      = A \odot z^{\circ} (z^{\circ})^{\top} \odot \Tilde{z} \Tilde{z}^{\top} \,.
  \end{equation*}
  By degree pruning condition (1), we have
  \begin{equation*}
      (A^* \one)_i \leq \log(1/\eta) R d^{\circ} \leq \sigma \hat{d} \,.
  \end{equation*}
  for all $i \in [n]$.
  
  Thus, $(A^*, z^*)$ is a feasible solution to $\cD(Y,z;A,\gamma,\sigma, \hat{d})$ with $\gamma = 2\eta$, $\sigma = 10 \log(1/\eta) R$ and $\hat{d} \geq \frac{1}{2} d^{\circ}$.
\end{proof}

%% file: fine-estimation-inhomo-robust-algorithm.tex
\subsection{Robust algorithm}
\label{sec:robust_fine_estimation_inhomo}

In this section, we show that the following algorithm based on sum-of-squares proofs in \cref{sec:fine-inhomo-sos} obtains a robust approximation of $\dnull$ that is optimal up to logarithmic factors.

\begin{algorithmbox}[Robust fine estimation algorithm for inhomogeneous random graphs]
    \label{alg:robust-fine-inhomo}
    \mbox{}\\
    \textbf{Input:}
        $\eta$-corrupted adjacency matrix $A$, corruption fraction $\eta$ and parameter $R$.

    \noindent
    \textbf{Algorithm:}
    \begin{enumerate}
        \item Obtain coarse estimator $\hat{d}$ by applying \cref{alg:robust-coarse} with $A, \eta, R$ as input.
        \item  Obtain level-$8$ pseudo-expectation $\tilde{\E}$ by solving sum-of-squares relaxation of program $\cD(Y,z;A,\gamma,\sigma, \hat{d})$ (defined in \cref{poly_sys:fine_inhomo_subgraph_intersect_and_max_degree_bound}) with $A$, $\gamma = 2\eta$, $\sigma = 10 \log(1/\eta) R$ and $\hat{d}$.
    \end{enumerate}
    
    \noindent
    \textbf{Output:} $\tilde{\E}[d(Y)]$
\end{algorithmbox}

\begin{theorem}[Robust fine estimation for inhomogeneous random graphs]
    \label{thm:robust-fine-estimation-inhomo}
    Let $\Qnull$ be an $n$-by-$n$ edge connection probability matrix and let $\dnull\seteq d(\Qnull)$.
    Suppose $\normi{\Qnull} \le R\dnull/n$ for some $R$.
    Let $A$ be an $\eta$-corrupted adjacency matrix of a random graph $\rv G^{\circ} \sim \bbG(n,Q^{\circ})$.
    Suppose $\eta\log(1/\eta)R \le c$ for some constant $c$ that is small enough.
    With probability $1-n^{-\Omega(1)}$, \cref{alg:robust-fine-inhomo} outputs an estimate $\tilde{d}$ satisfying
    \begin{equation*}
        \Abs{\frac{\tilde{d}}{\dnull}-1} \le 
        O \Paren{\sqrt{\frac{\log n}{\dnull n}} + \eta\log(1/\eta)R} \,.
    \end{equation*}
\end{theorem}

\begin{proof}
By \cref{thm:robust-coarse-estimation}, we have $\frac{1}{2}d^{\circ} \leq \hat{d} \leq \frac{3}{2} d^{\circ}$. Let $\gamma^* = 2\eta$, by \cref{lem:fine-inhomo-sos} and \cref{lem:fine-inhomo-existence-proof}, it follows that
  \begin{align*}
    \cD(Y,z;A,\gamma,\sigma, \hat{d})
    \sststile{8}{Y, z}
    |d(Y) - d(A^*)| & \leq 2 (\gamma + \gamma^*) \sigma \hat{d} \\
    & \leq 2 \cdot 4 \eta \cdot 10 \log(1/\eta) R \cdot \frac{3}{2} d^{\circ} \\
    & = 120 \eta \log(1/\eta) R d^{\circ} \,.
  \end{align*}
  and,
  \begin{equation*}
    |d(A^*) - d^{\circ}| \leq  10 \sqrt{\frac{d^{\circ} \log n}{n}} + 2 \eta \log(1/\eta) R d^{\circ} \,.
  \end{equation*}
  Therefore, we have 
    \begin{equation*}
    \cD(Y,z;A,\gamma,\sigma, \hat{d})
    \sststile{O(1)}{Y, z} \Abs{d(Y) - d^{\circ}} \leq 200 \eta \log(1/\eta) R d^{\circ} + 10 \sqrt{\frac{d^{\circ} \log n}{n}} \,.
  \end{equation*}
  Thus, the level-$8$ pseudo-expectation $\tilde{\E}$ satisfies
  \begin{equation*}
      \Abs{\tilde{\E}[d(Y)] - d^{\circ}} \leq 200 \eta \log(1/\eta) R d^{\circ} + 10 \sqrt{\frac{d^{\circ} \log n}{n}} \,,
  \end{equation*}
  which implies that
  \begin{equation*}
      \Abs{\frac{\tilde{\E}[d(Y)]}{d^{\circ}} - 1}
      \le O \Paren{\sqrt{\frac{\log n}{\dnull n}} + \eta\log(1/\eta)R} \,.
  \end{equation*}
\end{proof}

%% file: fine-estimation-inhomo-private-algorithm.tex
\subsection{Private algorithm}
\label{sec:private_fine_estimation_inhomo}

In this section, we present our algorithm and prove \cref{thm:fine_estimation_inhomo}. 
Our algorithm instantiates the sum-of-squares exponential mechanism in \cref{sec:sos_exp_mechanism}.

\paragraph{Score function}
For an $n$-by-$n$ symmetric matrix $A$ and a scalar $d$, we define the score of $d$ with regard to $A$ to be
\begin{align}
    \label{eq:fine_inhomo_sos_score}
    s(d; A) \seteq 
    \min_{0\le\gamma\le1} \gamma n \text{ s.t. } 
    \begin{cases} 
        \exists \text{ level-}8 \text{ pseudo-expectation } \tE \text{ satisfying } \\
        \cD(Y,z; A,\gamma,\sigma,\hat{d}) \,\cup\, \Set{\Abs{d(Y)-d} \le \alpha d} \,,
    \end{cases} 
\end{align}
where $\cD(Y,z; A,\gamma,\sigma,\hat{d})$ is the polynomial system defined in \cref{poly_sys:fine_inhomo_subgraph_intersect_and_max_degree_bound}, 
$\hat{d}$ is a coarse estimate, and
$\sigma, \alpha$ are fixed parameters whose values will be decided later.
Note that $(Y=\frac{d}{n}\one\one^\top, z=\zero)$ is a solution to the polynomial system $\cD(Y,z; A,1,\sigma,\hat{d}) \,\cup\, \Set{\Abs{d(Y)-d} \le \alpha d}$ for any $A\in\R^{n\times n}$ and any $d$ such that $0\le d\le \min\{\sigma\hat{d}, n\}$.

To efficiently compute $s(d;A)$, we can use the scheme as described in \cref{remark:score_function_computation}.

\paragraph{Exponential mechanism}
Given a privacy parameter $\eps$ and an $n$-by-$n$ symmetric matrix $A$, our private algorithm in \cref{thm:fine_estimation_inhomo} is the exponential mechanism with score function \cref{eq:fine_inhomo_sos_score} and range $[ 0, \min\set{\sigma\hat{d}, n} ]$.

\begin{algorithmbox}[Fine estimation for inhomogeneous random graphs]
    \label{alg:fine_inhomo_exp_mec}
    \mbox{}\\
    \textbf{Input:}
        Graph $A$, coarse estimate $\hat{d}$.
    
    \noindent
    \textbf{Parameters:} 
        $\eps, \sigma, \alpha$.
    
    \noindent
    \textbf{Output:} 
        A sample from the distribution $\mu_{A,\eps}$ with support $[ 0, \min\set{\sigma\hat{d}, n} ]$ and density 
        \begin{equation}
            \label{eq:fine_inhomo_exp_mec}
            \mathrm{d}\mu_{A,\eps}(d) \propto \exp\paren{-\eps\cdot s(d;A)} \,,
        \end{equation}
        where $s(d;A)$ is defined in \cref{eq:fine_inhomo_sos_score}.
\end{algorithmbox}

To efficiently sample from $\mu_{A,\eps}$, we can use the scheme as described in \cref{remark:sampling}.

\paragraph{Privacy}
The following privacy guarantee of our algorithm is a direct corollary of \cref{lem:sos_exp_mec_privacy}.

\begin{lemma}[Privacy]
    \label{lem:fine_inhomo_exp_mec_privacy}
    \cref{alg:fine_inhomo_exp_mec} is $2\eps$-differentially node private.
\end{lemma}

\paragraph{Utility}
The utility guarantee of our algorithm is stated in the following lemma.

\begin{lemma}[Utility]
    \label{lem:fine_inhomo_exp_mec_util}
    Let $\Qnull$ be an $n$-by-$n$ edge connection probability matrix and let $\dnull \seteq d(\Qnull)$. 
    Suppose $\Norm{\Qnull}_{\infty} \le R\dnull / n$ for some $R$.
    There is a sufficiently small constant $c$ such that the following holds.
    For any $\eta$ such that $\eta\log(1/\eta)R \le c$, given an $\eta$-corrupted inhomogeneous random graph $\bbG(n,\Qnull)$ and a coarse estimate $\hat{d}$ such that $0.5\dnull \le \hat{d} \le 2\dnull$,
    \cref{alg:fine_inhomo_exp_mec} outputs an estimate $\tilde{d}$ satisfying
    \[
        \Abs{\frac{\tilde{d}}{\dnull}-1} \le 
        O \Paren{\sqrt{\frac{\log n}{\dnull n}} +
        \frac{R \log^2 n}{\eps n} +
        \eta\log(1/\eta)R} \,,
    \]
    with probability $1-n^{-\Omega(1)}$.
\end{lemma}

Before proving \cref{lem:fine_inhomo_exp_mec_util}, we need the following two lemmas.

\begin{lemma}[Volume of low-score points]
    \label{lem:fine_inhomo_exp_mec_volume_of_low_score_points}
    Let $A\in\R^{n\times n}$ and $\eps>0$.
    Consider the distribution $\mu_{A,\eps}$ defined by \cref{eq:fine_inhomo_exp_mec}.
    Suppose $(Y=A^*, z=z^*)$ is a solution to $\cD(Y,z; A,\gamma^*,\sigma,\hat{d})$ and $2 \le d(A^*) \le \sigma\hat{d}$.
    Then for any $t\geq0$,
    \[
        \Pr_{ \rv{d} \sim \mu_{A,\eps} } \Paren{ s(\rv{d};A) \ge \gamma^* n + \frac{t\log n}{\eps} } 
        \le \frac{ n^{-t+1} }{\alpha} \,.
    \]
\end{lemma}
\begin{proof}
    Apply \cref{lem:sos_exp_mec_utility} with $\cD = [ 0,\min\{\sigma\hat{d}, n\} ]$ and
    \[
        \cG(A^*) = \Set{d\in\cD \,:\, \frac{d(A^*)}{1+\alpha} \le d \le \frac{d(A^*)}{1-\alpha}} \,.
    \]
    As $[d(A^*)/(1+\alpha), d(A^*)] \subseteq \cG(A^*)$ and $d(A^*) \ge 2 \ge 1+\alpha$, we have $\vol(\cG(A^*)) \ge \alpha$.
\end{proof}

\begin{lemma}[Low score implies utility]
    \label{lem:fine_inhomo_low_score_implies_utility}
    Let $A\in\R^{n\times n}$ and consider the score function $s(\cdot \,; A)$ defined in \cref{eq:fine_inhomo_sos_score}.
    Suppose $(Y=A^*, z=z^*)$ is a solution to $\cD(Y,z; A,\gamma^*,\sigma,\hat{d})$.
    For a scalar $d$ such that $s(d;A) \le \tau n$, 
    \[
       \frac{d(A^*) - 2(\gamma^*+\tau)\sigma\hat{d}}{1+\alpha} 
       \leq d \leq 
       \frac{d(A^*) + 2(\gamma^*+\tau)\sigma\hat{d}}{1-\alpha} \,.
    \]
\end{lemma}
\begin{proof}
    By \cref{lem:fine-inhomo-sos},
    \[
        \cD(Y,z; A,\tau,\sigma,\hat{d})
        \; \sststile{8}{Y, z} \;
        \Abs{d(Y) - d(A^*)} \leq 2 (\gamma^* + \tau) \sigma \hat{d} \,.
    \]
    Thus,
    \begin{align*}
        & \cD(Y,z; A,\tau,\sigma,\hat{d}) \,\cup\, \Set{\Abs{d(Y)-d} \le \alpha d} \\
        \sststile{8}{Y, z} \; 
        & \frac{d(A^*) - 2(\gamma^*+\tau)\sigma\hat{d}}{1+\alpha} \leq d \leq \frac{d(A^*) + 2(\gamma^*+\tau)\sigma\hat{d}}{1-\alpha} \,.
    \end{align*}
\end{proof}

Now we are ready to prove \cref{lem:fine_inhomo_exp_mec_util}.

\begin{proof}[Proof of \cref{lem:fine_inhomo_exp_mec_util}]
    Let $A$ be a realization of $\eta$-corrupted $\bbG(n,\Qnull)$.
    By \cref{lem:fine-inhomo-existence-proof}, the following event happens with probability at least $1-n^{-\Omega(1)}$.
    There exists a solution $(Y=A^*, z=z^*)$ to $\cD(Y,z; A,\gamma^*,\sigma,\hat{d})$ with 
    $\gamma^* = 2\eta$, 
    $\sigma = 10\log(1/\eta)R$, and
    \[
        \abs{d(A^*) - \dnull} \le 10 \sqrt{\dnull\log(n)/n} + 2 \eta \log(1/\eta) R\dnull \,.
    \]
    
    For $\eta \log(1/\eta) R$ smaller than some constant, we have $0.9\dnull \le d(A^*) \le 1.1\dnull$. 
    Note that $d(A^*) \ge 0.9\dnull \ge 2$ and $d(A^*) \le 1.1\dnull \le \sigma\hat{d}$.
    Then it follows by setting $t=10$ and $\alpha=n^{-2}$ in \cref{lem:fine_inhomo_exp_mec_volume_of_low_score_points} that, 
    \[
        \Pr_{\rv{d} \sim \mu_{A,\eps}} \Paren{s(\rv{d};A) \le \tau n} \ge 1 - n^{-7} 
        \text{ where } \tau \seteq 2\eta + 10\log(n)/(\eps n) \,.
    \]

    Let $\tilde d$ be a scalar such that $s(\tilde{d}; A) \le \tau n$.
    Then by \cref{lem:fine_inhomo_low_score_implies_utility}, 
    \[
       \frac{d(A^*) - 2(2\eta+\tau)\sigma\hat{d}}{1+\alpha} 
       \leq \tilde{d} \leq 
       \frac{d(A^*) + 2(2\eta+\tau)\sigma\hat{d}}{1-\alpha} \,.
    \]
    Plugging in everything, we have
    \[
        \Abs{\frac{\tilde d}{\dnull}-1} \le 
        O \Paren{\sqrt{\frac{\log n}{\dnull n}} +
        \frac{R\log(1/\eta)\log n}{\eps n} +
        R\eta\log(1/\eta)} \,.
    \]
    As an $\eta$-corrupted graph is actually uncorrupted when $\eta<1/n$, we can assume $\eta\ge1/(2n)$ without loss of generality.
    Therefore,
    \[
        \Abs{\frac{\tilde d}{\dnull}-1} \le 
        O\Paren{\sqrt{\frac{\log n}{\dnull n}} +
        \frac{R \log^2 n}{\eps n} +
        R\eta\log(1/\eta)} \,.
    \]
\end{proof}

\paragraph{Proof of \cref{thm:fine_estimation_inhomo}}
By \cref{lem:fine_inhomo_exp_mec_privacy} and \cref{lem:fine_inhomo_exp_mec_util}.

%% file: fine-estimation-ER.tex
\section{Fine estimation for \ER random graphs}
\label{sec:fine_estimation_ER}

From \cref{sec:coarse_estimation}, we have a a constant multiplicative approximation of the expected average degree $\dnull$.
In this section, we show how to use this coarse estimate to obtain our fine estimate for \ER random graphs.

\begin{theorem}[Fine estimation for \ER random graphs]
    \label{thm:fine_estimation_er}
    There are constants $C_1, C_2, C_3$ such that the following holds.
    For any $\eta\le C_1$, $\eps \ge C_2\log(n)/n$, and $\dnull\geq C_3$, there exists a polynomial-time $\eps$-differentially node private algorithm which, given an $\eta$-corrupted \ER random graph $\bbG(n,\dnull/n)$ and a constant-factor approximation of $\dnull$, outputs an estimate $\tilde{d}$ satisfying
    \[
        \Abs{\frac{\tilde d}{\dnull}-1} 
        \le 
        O \Paren{\sqrt{\frac{\log n}{\dnull n}} +
        \frac{\log^2 n}{\sqrt{\dnull} \eps n} +
        \frac{\eta\log n}{\sqrt{\dnull}}} \,,
    \]
    with probability $1-n^{-\Omega(1)}$.
\end{theorem}

We make a few remarks on \cref{thm:fine_estimation_er}.
\begin{itemize}
    \item Our algorithm in \cref{thm:fine_estimation_er} is an sum-of-squares exponential mechanism. $R,\eta,\eps$ are parameters given as input to our algorithm.

    \item We can get an estimate of $p^{\circ}$ by taking $\hat{p} = \frac{\hat{d}}{n-1}$. Since $\frac{\hat{p}}{\pnull} = \frac{\hat{d}}{\dnull}$, it follows that
    \[
        \Abs{\frac{\hat{p}}{\pnull}-1} 
        \le 
        O \Paren{\sqrt{\frac{\log n}{\dnull n}} +
        \frac{\log^2 n}{\sqrt{\dnull} \eps n} +
        \frac{\eta\log n}{\sqrt{\dnull}}} \,.
    \]

    \item Combining \cref{thm:coarse_estimation_inhomo} and \cref{thm:fine_estimation_er} gives us an efficient, private, and robust edge density estimation algorithm for \ER random graphs whose utility guarantee is information-theoretically optimal up to a factor of  $\log n$.
\end{itemize}

In \cref{sec:fine-er-sos}, we set up polynomial systems that our algorithm uses and prove useful sos inequalities.
In \cref{sec:robust_fine_estimation_er}, we show that we can easily obtain a robust algorithm via sos proofs in \cref{sec:fine-er-sos}.
Then in \cref{sec:private_fine_estimation_er}, we describe our algorithm and prove \cref{thm:fine_estimation_er}.

\input{fine-estimation-ER-sos}

\input{fine-estimation-ER-robust-algorithm}

\input{fine-estimation-ER-private-algorithm}

%% file: fine-estimation-ER-sos.tex
\subsection{Sum-of-squares}
\label{sec:fine-er-sos}

For an adjacency matrix $A$ and nonnegative scalars $\gamma$, $\sigma$ and $\hat{d}$, consider the following polynomial systems with indeterminates $Y=(Y_{ij})_{i,j\in[n]}$, $z=(z_i)_{i\in[n]}$ and coefficients that depend on $A,\gamma,\sigma,\delta,\hat{d}$:
\begin{gather}
  \cP_1(Y,z; A,\gamma) \seteq 
  \left. 
    \begin{cases}
      z \odot z = z,\, \iprod{\one, z} \ge (1-\gamma)n \\
      0 \leq Y \leq \one \one^{\top},\, Y=Y^\top \\
      Y \odot zz^\top = A \odot zz^\top
    \end{cases}
  \right\} \,, \label{poly_sys:fine_er_subgraph_intersect}
  \\
  \cP_4(Y; \sigma, \delta, \hat{d}) \seteq
  \left. 
    \begin{cases}
      |(Y\one)_i - d(Y)| \le \sigma \sqrt{\hat{d}} & \forall i\in[n] \\
      \Normop{Y - \frac{d(Y)}{n} \one \one^{\top}} \leq \delta \sqrt{\hat{d}}
    \end{cases}
  \right\} \,. \label{poly_sys:fine_er_max_degree_bound}
\end{gather}
For convenience of notation, we will consider the following combined polynomial system in remaining of the section
\begin{equation}
    \label{poly_sys:fine_er_subgraph_intersect_and_max_degree_bound}
    \cE(Y,z;A,\gamma,\sigma, \delta, \hat{d}) \seteq \cP_1(Y,z; A,\gamma) \,\cup\, \cP_4(Y; \sigma, \delta, \hat{d}) \,.
\end{equation}

\begin{lemma}
\label{lem:fine-er-sos}
    If $(A^*, z^*)$ is a feasible solution to $\cE(Y,z;A,\gamma^*,\sigma,\delta,\hat{d})$ and $\gamma + \gamma^* < 1$, then it follows that
    \begin{equation*}
        \cE(Y,z;A,\gamma,\sigma, \delta, \hat{d})
        \sststile{8}{Y, z}
        \Abs{d(Y) - d(A^*)}
        \leq \frac{4 (\gamma+\gamma^*) \sigma \sqrt{\hat{d}} + 2 (\gamma+\gamma^*) \delta \sqrt{\hat{d}}}{(1-\gamma-\gamma^*)^2} \,.
    \end{equation*}
\end{lemma}

\begin{proof}
Let $w = z \odot z^*$. Notice that, by \cref{lem:preliminary-SOS-set-union-bound}, we have $\cE \sststile{4}{z} 1-w_i \leq 2-z_i-z^*_i$ for all $i \in [n]$. Moreover, using similar analysis as in the proof of \cref{lem:coarse-sos}, it follows that
\begin{equation*}
    \cE
    \sststile{4}{Y, z}
    Y \odot w w^{\top} = A^* \odot w w^{\top} \,.
\end{equation*}
Therefore, we can get
\begin{align}
\label{eq:fine-er-sos-1}
\begin{split}
    \cE
    \sststile{4}{Y, z}
    n \Bigparen{d(Y) - d(A^*)} = & \iprod{Y-A^*,\one \one^{\top}} \\
    = & \iprod{Y-A^*,\one \one^{\top} - ww^{\top}} \\
    = & \iprod{Y-\frac{d(Y)}{n}\one \one^{\top} +\frac{d(Y)}{n}\one \one^{\top}-\frac{d(A^*)}{n}\one \one^{\top}+\frac{d(A^*)}{n}\one \one^{\top}-A^*,\one \one^{\top} - ww^{\top}} \\
    = & \iprod{Y-\frac{d(Y)}{n}\one \one^{\top},\one \one^{\top} - ww^{\top}} + \iprod{\frac{d(A^*)}{n}\one \one^{\top}-A^*,\one \one^{\top} - ww^{\top}} \\
    & + \iprod{\frac{d(Y)}{n}\one \one^{\top}-\frac{d(A^*)}{n}\one \one^{\top},\one \one^{\top} - ww^{\top}} \\
    = & \iprod{Y-\frac{d(Y)}{n}\one \one^{\top},\one \one^{\top} - ww^{\top}} + \iprod{\frac{d(A^*)}{n}\one \one^{\top}-A^*,\one \one^{\top} - ww^{\top}} \\
    & + \Bigparen{d(Y) - d(A^*)} \Paren{n - \frac{1}{n}\iprod{\one, w}^2} \,.
\end{split}
\end{align}
By rearranging terms, we can get
\begin{equation}
\label{eq:fine-er-sos-2}
    \cE
    \sststile{8}{Y, z}
    \frac{\iprod{\one, w}^2}{n} \Bigparen{d(Y) - d(A^*)}
    = \iprod{Y-\frac{d(Y)}{n}\one \one^{\top},\one \one^{\top} - ww^{\top}} + \iprod{\frac{d(A^*)}{n}\one \one^{\top}-A^*,\one \one^{\top} - ww^{\top}} \,.
\end{equation}
We bound the two terms on the right-hand side separately. For the first term $\iprod{Y-\frac{d(Y)}{n}\one \one^{\top},\one \one^{\top} - ww^{\top}}$, we have
\begin{equation}
\label{eq:fine-er-sos-3}
    \cE
    \sststile{8}{Y, z}
    \iprod{Y-\frac{d(Y)}{n}\one \one^{\top},\one \one^{\top} - ww^{\top}}
    = 2 \iprod{Y-\frac{d(Y)}{n}\one \one^{\top},\one (\one - w)^{\top}} + \iprod{\frac{d(Y)}{n}\one \one^{\top}-Y,(\one - w) (\one - w)^{\top}} \,.
\end{equation}
From constraints $|(Y\one)_i - d(Y)| \le \sigma \sqrt{\hat{d}}$ for all  $i\in[n]$, $\iprod{\one, z} \ge (1-\gamma)n$ and $\iprod{\one, z^*} \ge (1-\gamma^*)n$, we have
\begin{align}
\label{eq:fine-er-sos-4}
\begin{split}
    \cE
    \sststile{8}{Y, z}
    \iprod{Y-\frac{d(Y)}{n}\one \one^{\top},\one (\one - w)^{\top}}
    & = \iprod{Y \one-d(Y)\one,\one - w} \\
    & \leq \sum_{i \in [n]} (1-w_i) \sigma \sqrt{\hat{d}} \\
    & \leq \sigma \sqrt{\hat{d}} \cdot \Paren{\sum_{i \in [n]} 2-z_i-z^*_i} \\
    & \leq (\gamma+\gamma^*) n \sigma \sqrt{\hat{d}} \,.
\end{split}
\end{align}
From constraints $\Normop{Y - \frac{d(Y)}{n} \one \one^{\top}} \leq \delta \sqrt{\hat{d}}$, $\iprod{\one, z} \ge (1-\gamma)n$ and $\iprod{\one, z^*} \ge (1-\gamma^*)n$, we have
\begin{align}
\label{eq:fine-er-sos-5}
\begin{split}
    \cE
    \sststile{8}{Y, z}
    \iprod{\frac{d(Y)}{n}\one \one^{\top}-Y,(\one - w) (\one - w)^{\top}}
    & \leq \Normop{Y - \frac{d(Y)}{n} \one \one^{\top}} \Norm{\one - w}_2^2 \\
    & \leq \delta \sqrt{\hat{d}} \cdot \Paren{\sum_{i \in [n]} (1-w_i)^2} \\
    & = \delta \sqrt{\hat{d}} \cdot \Paren{\sum_{i \in [n]} 1-w_i} \\
    & \leq \delta \sqrt{\hat{d}} \cdot \Paren{\sum_{i \in [n]} 2-z_i-z^*_i} \\
    & \leq (\gamma+\gamma^*) n \delta \sqrt{\hat{d}} \,,
\end{split}
\end{align}
where the equality is because $\cE \sststile{2}{z} (1-w_i)^2 = (1-z_i z^*_i)^2 = 1-z_i z^*_i = 1-w_i$.

Plugging \cref{eq:fine-er-sos-4} and \cref{eq:fine-er-sos-5} into \cref{eq:fine-er-sos-3}, it follows that
\begin{equation}
\label{eq:fine-er-sos-6}
    \cE
    \sststile{8}{Y, z}
    \iprod{Y-\frac{d(Y)}{n}\one \one^{\top},\one \one^{\top} - ww^{\top}}
    \leq 2 (\gamma+\gamma^*) n \sigma \sqrt{\hat{d}} + (\gamma+\gamma^*) n \delta \sqrt{\hat{d}} \,.
\end{equation}
For the second term $\iprod{\frac{d(A^*)}{n}\one \one^{\top}-A^*,\one \one^{\top} - ww^{\top}}$, we can apply the same proof as above to get
\begin{equation}
\label{eq:fine-er-sos-7}
    \cE
    \sststile{8}{Y, z}
    \iprod{\frac{d(A^*)}{n}\one \one^{\top}-A^*,\one \one^{\top} - ww^{\top}}
    \leq 2 (\gamma+\gamma^*) n \sigma \sqrt{\hat{d}} + (\gamma+\gamma^*) n \delta \sqrt{\hat{d}} \,.
\end{equation}
Plugging \cref{eq:fine-er-sos-6} and \cref{eq:fine-er-sos-7} into \cref{eq:fine-er-sos-2}, it follows that
\begin{equation*}
    \cE
    \sststile{8}{Y, z}
    \frac{\iprod{\one, w}^2}{n} \Bigparen{d(Y) - d(A^*)}
    \leq 4 (\gamma+\gamma^*) n \sigma \sqrt{\hat{d}} + 2 (\gamma+\gamma^*) n \delta \sqrt{\hat{d}} \,.
\end{equation*}
Using the same proof strategy, we can also get
\begin{equation*}
    \cE
    \sststile{8}{Y, z}
    \frac{\iprod{\one, w}^2}{n} \Bigparen{d(A^*)-d(Y)}
    \leq 4 (\gamma+\gamma^*) n \sigma \sqrt{\hat{d}} + 2 (\gamma+\gamma^*) n \delta \sqrt{\hat{d}} \,.
\end{equation*}
Applying \cref{lem:preliminary-SOS-abs-to-square}, it follows that
\begin{equation}
\label{eq:fine-er-sos-8}
    \cE
    \sststile{8}{Y, z}
    \frac{\iprod{\one, w}^4}{n^2} \Bigparen{d(Y) - d(A^*)}^2
    \leq \Paren{4 (\gamma+\gamma^*) n \sigma \sqrt{\hat{d}} + 2 (\gamma+\gamma^*) n \delta \sqrt{\hat{d}}}^2 \,.
\end{equation}
Now, we would like to lower bound $\iprod{\one, w}^4$. By \cref{lem:preliminary-SOS-set-union-bound}, we have $\cE \sststile{4}{z} w_i \geq z_i+z^*_i-1$ for all $i \in [n]$. Therefore,
\begin{equation*}
    \cE
    \sststile{8}{Y, z}
    \iprod{\one, w}
    = \sum_{i \in [n]} w_i
    \geq \sum_{i \in [n]} (z_i+z^*_i-1)
    \geq (1- \gamma -\gamma^*) n \,.
\end{equation*}
Since $\gamma + \gamma^* < 1$, we have $1- \gamma -\gamma^* > 0$, and, therefore,
\begin{equation*}
    \cE
    \sststile{8}{Y, z}
    \iprod{\one, w}^4
    \geq (1- \gamma -\gamma^*)^4 n^4 \,.
\end{equation*}
Plugging this into \cref{eq:fine-er-sos-8}, we have
\begin{align*}
    \cE
    & \sststile{8}{Y, z}
    (1- \gamma -\gamma^*)^4 n^2 \Bigparen{d(Y) - d(A^*)}^2
    \leq \Paren{4 (\gamma+\gamma^*) n \sigma \sqrt{\hat{d}} + 2 (\gamma+\gamma^*) n \delta \sqrt{\hat{d}}}^2 \\
    & \sststile{8}{Y, z}
    \Bigparen{d(Y) - d(A^*)}^2
    \leq \frac{\Paren{4 (\gamma+\gamma^*) \sigma \sqrt{\hat{d}} + 2 (\gamma+\gamma^*) \delta \sqrt{\hat{d}}}^2}{(1- \gamma -\gamma^*)^4} \,.
\end{align*}
Applying \cref{lem:preliminary-SOS-square-to-abs}, it follows that
\begin{equation*}
    \cE
    \sststile{8}{Y, z}
    \Abs{d(Y) - d(A^*)}
    \leq \frac{4 (\gamma+\gamma^*) \sigma \sqrt{\hat{d}} + 2 (\gamma+\gamma^*) \delta \sqrt{\hat{d}}}{(1-\gamma-\gamma^*)^2} \,.
\end{equation*}
\end{proof}

\begin{lemma}
\label{lem:fine-er-existence-proof}
  Let $A$ be an $\eta$-corrupted adjacency matrix of a random graph $\rv G^{\circ} \sim \bbG(n, \frac{d^{\circ}}{n})$. With probability $1-n^{-\Omega(1)}$, there exists $A^*$ and $z^*$ such that
    \begin{itemize}
        \item $|d(A^*) - d^{\circ}| \leq 10 \sqrt{\frac{d^{\circ} \log n}{n}}$.
        \item $(A^*, z^*)$ is a feasible solution to $\cE(Y,z;A,\gamma,\sigma, \delta, \hat{d})$ with $\gamma = \eta$, $\sigma = 4 \log n$, $\delta = 4 C \sqrt{\log n}$ for some constant $C$ and $\hat{d} \geq \frac{1}{2} d^{\circ}$.
    \end{itemize}
\end{lemma}

\begin{proof}
  Let $\rv A^{\circ}$ be the adjacency matrix of $\rv G^{\circ}$ and $z^{\circ} \in \set{0, 1}^n$ denote the set of uncorrupted nodes ($z^{\circ}_i = 1$ if and only if node $i$ is uncorrupted).We will show that $A^* = \rv A^{\circ}$ and $z^* = z^{\circ}$ satisfies the lemma.
  
  \paragraph{Guarantee 1}
  By \cref{lem:average_degree_concentration}, we know that, with probability $1-n^{-\Omega(1)}$,
  \begin{equation}
  \label{eq:fine-er-robustness-1}
      |d(A^{\circ}) - d^{\circ}| \leq 10 \sqrt{\frac{d^{\circ} \log n}{n}} \,.
  \end{equation}
  
  \paragraph{Guarantee 2}
  It is easy to check that $z^* \odot z^* = z^*$, $0 \leq A^* \leq \one \one^{\top}$, $A^* = (A^*)^\top$ and $\iprod{\one, z^*} \geq 1 - \eta n$.
  By \cref{lem:degree_distribution}, we know that, with probability $1-n^{-\Omega(1)}$,
  \begin{equation}
  \label{eq:fine-er-robustness-2}
      \Norm{A^{\circ} \one - d^{\circ}\one}_{\infty} \leq \sqrt{d^{\circ}} \log n \,.
  \end{equation}
  Combining \cref{eq:fine-er-robustness-1} and \cref{eq:fine-er-robustness-2}, we have
  \begin{align}
  \label{eq:fine-er-robustness-3}
  \begin{split}
      \Norm{A^{\circ} \one - d(A^{\circ}) \one}_{\infty}
      & \leq \Norm{A^{\circ} \one - d^{\circ} \one}_{\infty} + \Norm{d^{\circ} \one - d(A^{\circ}) \one}_{\infty}\\
      & \leq \sqrt{d^{\circ}} \log n + 10 \sqrt{\frac{d^{\circ} \log n}{n}} \\
      & \leq 2 \log n \sqrt{d^{\circ}} \,.
  \end{split}
  \end{align}
  Therefore, for $\sigma = 4 \log n$ and $\hat{d} \geq \frac{1}{2}d^{\circ}$, it follows that
  \begin{equation*}
    \Abs{(A^* \one)_i - d(A^*)} \leq 2 \log n \sqrt{d^{\circ}} \leq \sigma \sqrt{\hat{d}} \,,
  \end{equation*}
  for all $i \in [n]$.
  
  By \cref{lem:spectral_bound_with_degree_pruning}, we know that, with probability $1-n^{-\Omega(1)}$, for some universal constant $C$,
  \begin{equation}
  \label{eq:fine-er-robustness-4}
      \Normop{A^{\circ} - \frac{d^{\circ}}{n}\one\one^\top} \leq C \sqrt{d^{\circ}\log n} \,.
  \end{equation}
  Combining \cref{eq:fine-er-robustness-1} and \cref{eq:fine-er-robustness-4}, we have
  \begin{align}
  \label{eq:fine-er-robustness-5}
  \begin{split}
      \Normop{A^{\circ} - \frac{d(A^{\circ})}{n}\one\one^\top}
      & \leq \Normop{A^{\circ} - \frac{d(A^{\circ})}{n}\one\one^\top} + \Normop{\frac{d(A^{\circ})}{n}\one\one^\top - \frac{d^{\circ}}{n}\one\one^\top}\\
      & \leq C \sqrt{d^{\circ}\log n} + 10 \sqrt{\frac{d^{\circ} \log n}{n}} \\
      & \leq 2 C \sqrt{d^{\circ}\log n} \,.
  \end{split}
  \end{align}
  Therefore, for $\delta = 4 C \sqrt{\log n}$ and $\hat{d} \geq \frac{1}{2} d^{\circ}$, it follows that
  \begin{equation*}
      \Normop{A^* - \frac{d(A^*)}{n}\one\one^\top} \leq 2 C \sqrt{d^{\circ}\log n} \leq \delta \sqrt{\hat{d}} \,.
  \end{equation*}
  Thus, $(A^*, z^*)$ is a feasible solution to $\cE(Y,z;A,\gamma,\sigma, \delta, \hat{d})$ with $\gamma = \eta$, $\sigma = 4 \log n$, $\delta = 4 C \sqrt{\log n}$ and $\hat{d} \geq \frac{1}{2} d^{\circ}$.
\end{proof}

%% file: fine-estimation-ER-robust-algorithm.tex
\subsection{Robust algorithm}
\label{sec:robust_fine_estimation_er}

In this section, we show that the following algorithm based on sum-of-squares proofs in \cref{sec:fine-er-sos} obtains a robust approximation of $\dnull$ that is optimal up to logarithmic factors.

\begin{algorithmbox}[Robust fine estimation algorithm for \ER random graphs]
    \label{alg:robust-fine-er}
    \mbox{}\\
    \textbf{Input:}
        $\eta$-corrupted adjacency matrix $A$ and corruption fraction $\eta$.

    \noindent
    \textbf{Algorithm:}
    \begin{enumerate}
        \item Obtain coarse estimator $\hat{d}$ by applying \cref{alg:robust-coarse} with $A, \eta, R=1$ as input.
        \item  Obtain level-$8$ pseudo-expectation $\tilde{\E}$ by solving sum-of-squares relaxation of program $\cE(Y,z;A,\gamma,\sigma, \delta, \hat{d})$ (defined in \cref{poly_sys:fine_er_subgraph_intersect_and_max_degree_bound}) with $A$, $\gamma = \eta$, $\sigma = 4 \log n$, $\delta = 4 C \sqrt{\log n}$ and $\hat{d}$.
    \end{enumerate}
    
    \noindent
    \textbf{Output:} $\tilde{\E}[d(Y)]$
\end{algorithmbox}

\begin{theorem}[Robust fine estimation for \ER random graphs]
    \label{thm:robust-fine-estimation-er}
    Let $A$ be an $\eta$-corrupted adjacency matrix of a random graph $\rv G^{\circ} \sim \bbG(n, \frac{d^{\circ}}{n})$.
    With probability $1-n^{-\Omega(1)}$, \cref{alg:robust-fine-er} outputs an estimate $\tilde{d}$ satisfying
    \begin{equation*}
        \Abs{\frac{\tilde d}{\dnull}-1} 
        \le O \Paren{\sqrt{\frac{\log n}{\dnull n}} + \frac{\eta\log n}{\sqrt{\dnull}}} \,.
    \end{equation*}
\end{theorem}

\begin{proof}
By \cref{thm:robust-coarse-estimation}, we have $\frac{1}{2}d^{\circ} \leq \hat{d} \leq \frac{3}{2} d^{\circ}$. Let $\gamma^*=\eta$, by \cref{lem:fine-er-sos} and \cref{lem:fine-er-existence-proof}, it follows that
\begin{align*}
    \cE(Y,z;A,\gamma,\sigma, \delta, \hat{d})
    \sststile{O(1)}{Y, z}
    \Abs{d(Y) - d(A^*)}
    & \leq \frac{4 (\gamma+\gamma^*) \sigma \sqrt{\hat{d}} + 2 (\gamma+\gamma^*) \delta \sqrt{\hat{d}}}{(1-\gamma-\gamma^*)^2} \\
    & \leq \frac{40 \eta \log n \sqrt{d^{\circ}} + 40 C \eta \sqrt{d^{\circ} \log n}}{(1-2\eta)^2} \\
    & \leq C' \eta \log n \sqrt{d^{\circ}} \,,
\end{align*}
for some constant $C'$, and,
\begin{equation*}
    |d(A^*) - d^{\circ}| \leq 10 \sqrt{\frac{d^{\circ} \log n}{n}}\,.
\end{equation*}
Therefore, we have 
\begin{align*}
    \cE(Y,z;A,\gamma,\sigma, \delta, \hat{d})
    \sststile{O(1)}{Y, z} \Abs{d(Y) - d^{\circ}} \leq C' \eta \log n \sqrt{d^{\circ}} + 10 \sqrt{\frac{d^{\circ} \log n}{n}} \,.
\end{align*}
Thus, the level-$8$ pseudo-expectation $\tilde{\E}$ satisfies
\begin{equation*}
  \Abs{\tilde{\E}[d(Y)] - d^{\circ}} \leq C' \eta \log n \sqrt{d^{\circ}} + 10 \sqrt{\frac{d^{\circ} \log n}{n}} \,,
\end{equation*}
which implies that
\begin{equation*}
  \Abs{\frac{\tilde{\E}[d(Y)]}{d^{\circ}} - 1}
  O \Paren{\sqrt{\frac{\log n}{\dnull n}} + \frac{\eta\log n}{\sqrt{\dnull}}} \,.
\end{equation*}
\end{proof}

%% file: fine-estimation-ER-private-algorithm.tex
\subsection{Private algorithm}
\label{sec:private_fine_estimation_er}

In this section, we present our algorithm and prove \cref{thm:fine_estimation_er}. 
Our algorithm instantiates the sum-of-squares exponential mechanism in \cref{sec:sos_exp_mechanism}.

\paragraph{Score function}
For an $n$-by-$n$ symmetric matrix $A$ and a scalar $d$, we define the score of $d$ with regard to $A$ to be
\begin{align}
    \label{eq:fine_er_sos_score}
    s(d; A) \seteq 
    \min_{0\le\gamma\le1} \gamma n \text{ s.t. } 
    \begin{cases} 
        \exists \text{ level-}8 \text{ pseudo-expectation } \tE \text{ satisfying } \\
        \cE(Y,z; A,\gamma,\sigma,\delta,\hat{d}) \,\cup\, \Set{\Abs{d(Y)-d} \le \alpha d} \,,
    \end{cases} 
\end{align}
where $\cE(Y,z; A,\gamma,\sigma,\delta,\hat{d})$ is the polynomial system defined in \cref{poly_sys:fine_er_subgraph_intersect_and_max_degree_bound}, 
$\hat{d}$ is a coarse estimate, and
$\sigma, \delta, \alpha$ are fixed parameters whose values will be decided later.
Note that $(Y=\frac{d}{n}\one\one^\top, z=\zero)$ is a solution to the polynomial system $\cE(Y,z; A,1,\sigma,\delta,\hat{d}) \,\cup\, \Set{\Abs{d(Y)/d-1}\le\alpha}$ for any $A\in\R^{n\times n}$ and any $d\in[0,n]$.

To efficiently compute $s(d;A)$, we can use the scheme as described in \cref{remark:score_function_computation}.

\paragraph{Exponential mechanism}
Given a privacy parameter $\eps$ and an $n$-by-$n$ symmetric matrix $A$, our private algorithm in \cref{thm:fine_estimation_er} is the exponential mechanism with score function \cref{eq:fine_er_sos_score} and range $[0,n]$.

\begin{algorithmbox}[Fine estimation for \ER random graphs]
    \label{alg:fine_er_exp_mec}
    \mbox{}\\
    \textbf{Input:}
        Graph $A$, coarse estimate $\hat{d}$.
    
    \noindent
    \textbf{Parameters:} 
        $\eps, \sigma, \delta, \alpha$.
    
    \noindent
    \textbf{Output:} 
        A sample from the distribution $\mu_{A,\eps}$ with support $[0,n]$ and density 
        \begin{equation}
            \label{eq:fine_er_exp_mec}
            \mathrm{d}\mu_{A,\eps}(d) \propto \exp\paren{-\eps\cdot s(d;A)} \,,
        \end{equation}
        where $s(d;A)$ is defined in \cref{eq:fine_er_sos_score}.
\end{algorithmbox}

To efficiently sample from $\mu_{A,\eps}$, we can use the scheme as described in \cref{remark:sampling}.

\paragraph{Privacy}
The following privacy guarantee of our algorithm is a direct corollary of \cref{lem:sos_exp_mec_privacy}.

\begin{lemma}[Privacy]
    \label{lem:fine_er_exp_mec_privacy}
    \cref{alg:fine_er_exp_mec} is $2\eps$-differentially node private.
\end{lemma}

\paragraph{Utility}
The utility guarantee of our algorithm is stated in the following lemma.

\begin{lemma}[Utility]
    \label{lem:fine_er_exp_mec_util}
    There are constants $C_1, C_2, C_3$ such that the following holds.
    For any $\eta\le C_1$, $\eps \ge C_2\log(n)/n$, and $\dnull\geq C_3$, 
    given an $\eta$-corrupted \ER random graph $\bbG(n, \dnull/n)$ and a coarse estimate $\hat{d}$ such that $0.5\dnull \le \hat{d} \le 2\dnull$,
    \cref{alg:fine_er_exp_mec} outputs an estimate $\tilde{d}$ satisfying
    \[
        \Abs{\frac{\tilde d}{\dnull}-1} 
        \le 
        O \Paren{\sqrt{\frac{\log n}{\dnull n}} +
        \frac{\log^2 n}{\sqrt{\dnull} \eps n} +
        \frac{\eta\log n}{\sqrt{\dnull}}} \,,
    \]
    with probability $1-n^{-\Omega(1)}$.
\end{lemma}

Before proving \cref{lem:fine_er_exp_mec_util}, we need the following two lemmas.

\begin{lemma}[Volume of low-score points]
    \label{lem:fine_er_exp_mec_volume_of_low_score_points}
    Let $A\in\R^{n\times n}$ and $\eps>0$.
    Consider the distribution $\mu_{A,\eps}$ defined by \cref{eq:fine_er_exp_mec}.
    Suppose $(Y=A^*, z=z^*)$ is a solution to $\cE(Y,z; A,\gamma^*)$ and $d(A^*) \ge 2$.
    Then for any $t\geq0$,
    \[
        \Pr_{\rv d\sim\mu_{A,\eps}}\Paren{s(\rv d;A) \ge \gamma^* n + \frac{t\log n}{\eps}} 
        \le \frac{n^{-t+1}}{\alpha} \,.
    \]
\end{lemma}
\begin{proof}
    Apply \cref{lem:sos_exp_mec_utility} with $\cD=[0,n]$ and
    \[
        \cG(A^*) = \Set{d\in\cD \,:\, \frac{d(A^*)}{1+\alpha} \le d \le \frac{d(A^*)}{1-\alpha}} \,.
    \]
    As $[d(A^*)/(1+\alpha), d(A^*)] \subseteq \cG(A^*)$ and $d(A^*)\ge2\ge1+\alpha$, we have $\vol(\cG(A^*)) \ge \alpha$.
\end{proof}

\begin{lemma}[Low score implies utility]
    \label{lem:fine_er_low_score_implies_utility}
    Let $A\in\R^{n\times n}$ and consider the score function $s(\cdot \,; A)$ defined in \cref{eq:fine_er_sos_score}.
    Suppose $(Y=A^*, z=z^*)$ is a solution to $\cE(Y,z; A,\gamma^*)$.
    For a scalar $d$ such that $s(d;A) \le \tau n$ and $\gamma^* + \tau \le 0.1$,
    \[
        \frac{1}{1+\alpha} \Paren{d(A^*) - 5(\gamma^* + \tau)(\sigma+\delta)\sqrt{\hat d}}  
        \le d \le
        \frac{1}{1-\alpha} \Paren{d(A^*) + 5(\gamma^* + \tau)(\sigma+\delta)\sqrt{\hat d}} \,.
    \]
\end{lemma}
\begin{proof}
    Applying \cref{lem:fine-er-sos} with $\gamma^* + \tau \le 0.1$, we have
    \[
        \cE(Y,z; A,\tau)
        \; \sststile{8}{Y, z} \;
        \Abs{d(Y) - d(A^*)} \leq 5(\gamma^* + \tau)(\sigma+\delta)\sqrt{\hat d} \,.
    \]
    Thus,
    \begin{align*}
        & \cE(Y,z; A,\tau) \,\cup\, \Set{\Abs{d(Y)-d}\le\alpha d} \\
        \sststile{8}{Y, z} \; 
        & \frac{1}{1+\alpha} \Paren{d(A^*) - 5(\gamma^* + \tau)(\sigma+\delta)\sqrt{\hat d}}  
        \le d \le
        \frac{1}{1-\alpha} \Paren{d(A^*) + 5(\gamma^* + \tau)(\sigma+\delta)\sqrt{\hat d}} \,.
    \end{align*}
\end{proof}

Now we are ready to prove \cref{lem:fine_er_exp_mec_util}.

\begin{proof}[Proof of \cref{lem:fine_er_exp_mec_util}]
    Let $A$ be a realization of $\eta$-corrupted $\bbG(n,\dnull/n)$.
    By \cref{lem:fine-er-existence-proof}, the following event happens with probability at least $1-n^{-\Omega(1)}$.
    There exists a solution $(Y=A^*, z=z^*)$ to $\cE(Y,z; A,\gamma^*,\sigma,\delta,\hat{d})$ where $\gamma^* = \eta$, $\sigma \le O(\log{n})$, $\delta \le O(\sqrt{\log n})$, and $\abs{d(A^*) - \dnull} \le  O \Paren{\sqrt{\dnull\log(n)/n}}$ .
    
    As $d(A^*) \ge 0.9\dnull \ge 2$, then it follows by setting $t=10$ and $\alpha=n^{-2}$ in \cref{lem:fine_er_exp_mec_volume_of_low_score_points} that, 
    \[
        \Pr_{\rv{d} \sim \mu_{A,\eps}} \Paren{s(\rv{d};A) \le \tau n} \ge 1 - n^{-7} 
        \text{ where } \tau \seteq 2\eta + 10\log(n)/(\eps n) \,.
    \]
    Let $\tilde d$ be a scalar such that $s(\tilde{d}; A)\le\tau n$.
    For $\eta$ and $\log(n)/(\eps n)$ smaller than some constant, we have $2\eta + \tau \le 0.1$.
    Then by \cref{lem:fine_er_low_score_implies_utility},
    \[
        \frac{1}{1+\alpha} \Paren{d(A^*) - 5(\eta + \tau)(\sigma+\delta)\sqrt{\hat d}}  
        \le \tilde{d} \le
        \frac{1}{1-\alpha} \Paren{d(A^*) + 5(\eta + \tau)(\sigma+\delta)\sqrt{\hat d}} \,.
    \]
    Plugging in everything, we have
    \[
        \Abs{\frac{\tilde d}{\dnull}-1} 
        \le 
        O \Paren{\sqrt{\frac{\log n}{\dnull n}} +
        \frac{\log^2 n}{\sqrt{\dnull} \eps n} +
        \frac{\eta\log n}{\sqrt{\dnull}}} \,.
    \]
\end{proof}

\paragraph{Proof of \cref{thm:fine_estimation_er}}
By \cref{lem:fine_er_exp_mec_privacy} and \cref{lem:fine_er_exp_mec_util}.

%% file: lower-bounds.tex
\section{Lower bounds}
\label{sec:lower-bounds}

In this section, we prove \cref{thm:lower_bound_ER}, \cref{thm:robust_lower_bound_inhomo}, and \cref{thm:privacy_lower_bound_inhomo}.

\input{lower-bound-er}
\input{lower-bound-inhomo}

%% file: lower-bound-er.tex
\subsection{Lower bound for \ER random graphs}

In this section, we prove \cref{thm:lower_bound_ER}.

\restatetheorem{thm:lower_bound_ER}

We leave the formal proof of \cref{thm:lower_bound_ER} to the end of this section. Now we sketch the proof idea.
One natural idea to prove this theorem is to construct a coupling $\omega$ of $\bbG(n,\pnull)$ and $\bbG(n,(1-2\alpha)\pnull)$ such that for $(\rv{G}, \rv{G}') \sim \omega$, the typical distance between $\rv G$ and $\rv G'$ can be well controlled.
However, such a coupling is tricky to construct directly, as the node degrees in an \ER random graph are not independent.
To avoid dealing with such dependence, we instead consider the directed \ER random graphs, which is inspired by the proof of \cite[Theorem 1.5]{acharya2022robust}.
The directed \ER random graph model, denoted by $\tilde{\bbG}(n,\pnull)$, is a distribution over $n$-node directed graphs where each edge $(i,j)$ is present with probability $\pnull$ independently.
Since the outdegrees in a directed \ER random graph are i.i.d. Binomial random variables, it is not so difficult to construct a coupling of $\tilde{\bbG}(n,\pnull)$ and $\tilde{\bbG}(n,(1-2\alpha)\pnull)$.
Then we can convert such a coupling into a coupling of $\bbG(n,\pnull)$ and $\bbG(n,(1-2\alpha)\pnull)$.

\begin{lemma}[Coupling]
    \label{lem:coupling}
    Let $\pnull \in [0,1]$, $\alpha \in [0,1/2]$, and $p' \seteq (1-2\alpha) \pnull$.
    There exists a coupling $\omega$ of $\bbG(n, \pnull)$ and $\bbG(n, p')$ with the following property.
    For $(\rv G, \rv G') \sim \omega$, the distribution of $\dist(\rv G, \rv G')$ is the binomial distribution $\Bin(n, \Delta)$ where $\Delta = \TV\paren{\Bin(n,\pnull), \Bin(n,p')}$.
    Moreover, if $\pnull \le c$ and $\alpha \le c' / \sqrt{n\pnull}$ for some constants $c, c'$, then $\Delta \lesssim \alpha \sqrt{n\pnull}$.
\end{lemma}
\begin{proof}
    We first show that it suffices to construct a coupling of $\tilde{\bbG}(n, \pnull)$ and $\tilde{\bbG}(n, p')$.
    For a directed graph $\tilde{G}$, it can be converted into an undirected graph $U(\tilde G)$ by letting $\{i,j\}\in U(\tilde G)$ iff $i\leq j$ and $(i,j)\in\tilde{G}$.
    It is easy to see that\footnote{
        For a directed graph $\tilde{G}$, we define its adjacency matrix $\tilde A$ to be $\tilde A(i,j) \seteq \one\{(i,j)\in\tilde{G}\}$.
        The (node) distance between two $n$-node directed graphs $\tilde{G}, \tilde{G}'$, denoted by $\dist(\tilde{G}, \tilde{G}')$, is number of nonzero rows of $\tilde{A}-\tilde{A}'$.
    }
    $\dist(\tilde{G} , \tilde{G}') = \dist(U(\tilde{G}) , U(\tilde{G}'))$.
    Also observe that if $\tilde{\rv G} \sim \tilde{\bbG}(n,\pnull)$ then $U(\tilde{\rv G}) \sim \bbG(n,\pnull)$.
    Therefore, a coupling $\tilde{\omega}$ of $\tilde{\bbG}(n, \pnull)$ and $\tilde{\bbG}(n, p')$ can be easily converted in to a coupling $\omega$ of ${\bbG}(n, \pnull)$ and ${\bbG}(n, p')$ such that, 
    for $(\tilde{\rv G}, \tilde{\rv G}') \sim \tilde{\omega}$ and $(\rv G, {\rv G}') \sim \omega$, 
    we have
    \[
        \dist\Paren{\tilde{\rv G},  \tilde{\rv G}'}
        \stackrel{\mathrm{d}}{=}
        \dist\Paren{\rv G, {\rv G}'} \,.
    \]
    
    Now we construct a coupling of $\tilde{\bbG}(n, \pnull)$ and $\tilde{\bbG}(n, p')$.
    Instead of sampling each edge independently, an equivalent way to sample from $\tilde\bbG(n,\pnull)$ is as follows. 
    For each $i\in[n]$ :
    \begin{itemize}
        \item Sample an outdegree $\rv{d}\sim\Bin(n,\pnull)$.
        \item Sample a uniformly random subset $\rv{S} \subseteq [n]$ of size $\rv{d}$. For each $j\in\rv{S}$, add an edge from $i$ to $j$.
    \end{itemize}
    Then it is easy to see there exists a coupling $\tilde\omega$ of $\tilde\bbG(n,\pnull)$ and $\tilde\bbG(n,p')$ such that if $(\tilde{\rv G}, \tilde{\rv G}') \sim \tilde\omega$ then $\dist(\tilde{\rv G}, \tilde{\rv G}') \sim \Bin(n,\Delta)$
    where 
    \[
        \Delta = \TV\Paren{\Bin(n,\pnull), \Bin(n,p')} \,.
    \]
    We have the following bound on the total variation between binomial distributions (see e.g. \cite[Equation (2.15)]{adell2006exact}). 
    For $0<p<1$ and $0<x<1-p$, 
    \[ 
        \TV\Paren{\Bin(N,p), \Bin(N,p+x)} 
        \le \frac{\sqrt{e}}{2} \frac{\tau(x)}{(1-\tau(x))^2} \,,
    \]
    where $\tau(x) \seteq x\sqrt{\frac{N+2}{2p(1-p)}}$ , provided $\tau(x)<1$.
    Plugging in $N=n$, $p=\pnull$, and $x=2\alpha \pnull$, we have 
    \[
        \Delta = \TV\Paren{\Bin(n,\pnull), \Bin(n,p')}
        \lesssim \alpha\sqrt{n \pnull} \,,
    \]
    provided $\pnull \leq c$ and $\alpha \leq c'/\sqrt{n\pnull}$ for sufficiently small absolute constants $c,c'$.
\end{proof}

\begin{proof}[Proof of \cref{thm:lower_bound_ER}]
    Let $\cA$ be an algorithm satisfying the theorem's assumptions.
    Let $p'\seteq(1-2\alpha)\pnull$.
    Let $\omega$ be a coupling of $\bbG(n, \pnull)$ and $\bbG(n, p')$ as guaranteed by \cref{lem:coupling}.
    Then for $(\rv G, \rv G') \sim \omega$, we have $\dist(\rv G, \rv G') \sim \Bin(n, \Delta)$ where $\Delta = \TV\paren{\Bin(n,\pnull), \Bin(n,p')}$.

    By the utility assumption of algorithm $\cA$,
    \[
        \Pr_{\cA, \bbG(n,\pnull)} \Paren{\Abs{\cA(\rv G) - \pnull} < \alpha \pnull} \geq 1-\beta \,. 
    \]

    As algorithm $\cA$ is $\eps$-DP, we have for any graphs $G,G'$ that,
    \[ 
        \Pr_{\cA} \Paren{ \Abs{\cA(G') - p'} < \alpha \pnull } 
        \leq 
        e^{\eps\cdot\dist(G,G')} \cdot \Pr_{\cA} \Paren{\Abs{\cA(G)-p'} < \alpha \pnull} \,.
    \]
    Taking expectation w.r.t. the coupling $\omega$ on both sides gives
    \begin{align*}
        \E_{\omega} \E_{\cA} \one\Set{\Abs{\cA(\rv{G}')-p'} < \alpha \pnull} 
        &\leq \E_{\omega} e^{\eps\cdot\dist(\rv{G},\rv{G}')} \cdot \E_{\cA} \one\Set{\Abs{\cA(\rv{G})-p'} < \alpha \pnull} \,, \\
        \Pr_{\cA,\tilde\bbG(n,p')} \Paren{\Abs{\cA(\rv{G}')-p'} < \alpha \pnull}
        &\leq \E_{\omega,\cA} e^{\eps\cdot\dist(\rv{G},\rv{G}')} \cdot \one\Set{\Abs{\cA(\rv{G})-p'} < \alpha \pnull} \,.  \numberthis \label{eq:lb-er-1}
    \end{align*}
    By the utility assumption of algorithm $\cA$ and $p'<\pnull$, the left-hand side of \cref{eq:lb-er-1} is at least $1-\beta$.
    Using the Cauchy-Schwartz inequality, the right-hand side of \cref{eq:lb-er-1} can be upper bounded as follows,
    \begin{align*}
        \E_{\omega,\cA} e^{\eps\cdot\dist(\rv{G},\rv{G}')} \cdot \one\Set{\Abs{\cA(\rv{G})-p'} < \alpha \pnull} 
        &\leq \sqrt{\E_{\omega,\cA} e^{2\eps\cdot\dist(\rv{G},\rv{G}')}} \sqrt{\E_{\omega,\cA} \one\Set{\Abs{\cA(\rv{G})-p'} < \alpha \pnull}} \\
        &\leq \sqrt{\E_{\Bin(n,\Delta)} e^{2\eps\cdot\rv X}} \sqrt{\Pr_{\cA,\tilde\bbG(n,\pnull)} \Paren{\Abs{\cA(\rv{G})-p'} < \alpha \pnull}} \,.
    \end{align*}
    By squaring both sides of \cref{eq:lb-er-1} and plugging in the above two bounds, we have
    \[
        (1-\beta)^2 \leq \E_{\Bin(n,\Delta)} \Brac{e^{2\eps\cdot\rv X}} \cdot \Pr_{\cA,\tilde\bbG(n,\pnull)} \Paren{\Abs{\cA(\rv{G})-p'} < \alpha \pnull} \,.
    \]
    Using the formula for the moment generating function of binomial distributions, we have
    \[
        \E_{\Bin(n,\Delta)} \Brac{e^{2\eps\cdot\rv X}}
        = \Paren{1 + \Delta(e^{2\eps}-1)}^n
        \leq e^{n\Delta(e^{2\eps}-1)} \,.
    \]
    Then
    \[
        \Pr_{\cA,\bbG(n,\pnull)} \Paren{\Abs{\cA(\rv{G})-p'} < \alpha \pnull}
        \ge (1-\beta)^2 \cdot e^{-n\Delta(e^{2\eps}-1)} \,.
    \]
    Since $p'-\pnull = 2\alpha \pnull$, the two events $\set{\hat{p} : \abs{\hat{p}-\pnull} < \alpha \pnull}$ and $\set{\hat{p} : \abs{\hat{p}-p'} < \alpha \pnull}$ are disjoint. 
    Thus,
    \[
        \Pr_{\cA,\bbG(n,\pnull)} \Paren{\Abs{\cA(\rv{G})-p'} < \alpha \pnull}
        \leq 1 - \Pr_{\cA,\bbG(n,\pnull)} \Paren{\Abs{\cA(\rv{G})-\pnull} < \alpha \pnull}
        \leq \beta \,.
    \]
    Therefore, we have the following lower bound
    \[
        \Delta \geq \frac{2\log(1-\beta)+\log(1/\beta)}{n(e^{2\eps}-1)} \,,
    \]
    which is $\Delta \gtrsim \frac{\log(1/\beta)}{n\eps}$ for samll enough $\eps$ and $\beta$.

    By \cref{lem:coupling}, if $\pnull \le c$ and $\alpha \le c' / \sqrt{n\pnull}$ for some constants $c, c'$, then $\Delta \lesssim \alpha \sqrt{n\pnull}$. Combined with the lower bound $\Delta \gtrsim \frac{\log(1/\beta)}{n\eps}$, we have
    \[ 
        \alpha \gtrsim \frac{\log(1/\beta)}{n\eps\sqrt{n\pnull}} \,.
    \]
\end{proof}

%% file: lower-bound-inhomo.tex
\subsection{Lower bound for inhomogeneous random graphs} 

In this section, we prove \cref{thm:robust_lower_bound_inhomo} and \cref{thm:privacy_lower_bound_inhomo}.

We first show the lower bound for the error rate of robust estimation.
\restatetheorem{thm:robust_lower_bound_inhomo}
\begin{proof}  
    Let $\pnull\in [0,1]$, and let $\Qnull\in [0,1]^{n\times n}$ be the matrix, in which all entries are $\pnull$, except for the rows and columns corresponding to a set of $\eta n$ nodes setting to be $R \pnull$. 
    Let $Q$ be the matrix, in which all entries are $\pnull$, except for the rows and columns corresponding to a set of $\eta n$ nodes setting to be $0$.
    
    We construct the following pair of distributions $\cD_0$ and $\cD_1$:
    \begin{itemize}
        \item $\cD_0$: The distribution of $\rv{G} \sim \bbG(\Qnull)$.
        \item $\cD_1$: The distribution of $\rv{G} \sim \bbG(Q)$.
    \end{itemize}
    Then we have $\frac{1}{n^2}\Abs{\norm{\Qnull}_1- \norm{Q}_1} \geq \Omega(R\eta n^2 \pnull)$.
    
    On the other hand, there is a coupling between  $\tilde{G}\sim \bbG(\Qnull)$ and $\tilde{G'}\sim \bbG(Q)$ such that $\text{dist}(\tilde{G},\tilde{G'})\leq \eta n$.
    Therefore, the two distributions are indistinguishable under the $\eta$-corruption model.
    Since the edge density of $\bbG(\Qnull)$ differs from $\bbG(Q)$ by $\Omega(R\eta \pnull)$, no algorithm can achieve error rate $o(R\eta \pnull)$ with probability $1-o(1)$ for both distributions under the corruption of $\eta$-fraction of the nodes.
\end{proof}

\restatetheorem{thm:privacy_lower_bound_inhomo}

\begin{proof}[Proof of \cref{thm:privacy_lower_bound_inhomo}]
  We will prove the lower bound by constructing a pair of distributions $\cD_0$ and $\cD_1$ such that the total variation distance between them is small, but the difference in edge density is significant.
  Then since $\eps$-differentially node-private algorithm needs to have similar distributions in the output, it could not succeed in accurately estimating the edge density accurately under both distributions. 

  Let $\eta\in [0,0.001)$.
  Let $\pnull\in [0,1]$, and let $\Qnull\in [0,1]^{n\times n}$ be the matrix, in which all entries are $\pnull$, except for the rows and columns corresponding to a set of $\eta n$ nodes setting to be $0$.
  Let $Q$ be the matrix, in which all entries are $\pnull$, except for the rows and columns corresponding to a set of $\eta n$ nodes setting to be $R\pnull$.
  
  We construct the following pair of distributions $\cD_0$ and $\cD_1$:
  \begin{itemize}
    \item $\cD_0$: The distribution of $\rv{G} \sim \bbG(\Qnull)$.
    \item $\cD_1$: The distribution of $\rv{G} \sim \bbG(Q)$.
  \end{itemize}
  Let $p'=\frac{\norm{\Qnull}_1}{n^2}$ and $p=\frac{\norm{Q}_1}{n^2}$.
  We have $\Abs{p-p'} \geq R\eta \pnull$.

  On the other hand, there is a coupling between  $\tilde{G}\sim \bbG(Q)$ and $\tilde{G'}\sim \bbG(\Qnull)$ such that $\text{dist}(\tilde{G},\tilde{G'})\leq \eta n$.
  Taking expectation w.r.t. the coupling $\omega$ on both sides gives
  \begin{align*}
    \E_{\omega} \E_{\cA} \one\Set{\Abs{\cA(\tilde{\rv G}')-p} < \frac{R\eta}{2} \pnull} 
    &\leq \E_{\omega} e^{\eps\cdot\dist(\tilde{\rv G},\tilde{\rv G}')} \cdot \E_{\cA} \one\Set{\Abs{\cA(\tilde{\rv G})-p} < \frac{R\eta}{2} \pnull} \,, \\
    \Pr_{\cA,\tilde\bbG(\Qnull)} \Paren{\Abs{\cA(\tilde{\rv G}')-p} < \frac{R\eta}{2} \pnull}
    &\leq \E_{\omega,\cA} e^{\eps\cdot\dist(\tilde{\rv G},\tilde{\rv G}')} \cdot \one\Set{\Abs{\cA(\tilde{\rv G})-p} < \frac{R\eta}{2} \pnull} \,.
  \end{align*}
  By the utility assumption of algorithm $\cA$ and $p<\pnull$, the left-hand side is at least $1-\beta$.
  Using the Cauchy-Schwartz inequality, the right-hand side can be upper bounded as follows,
  \begin{align*}
    \E_{\omega,\cA} e^{\eps\cdot\dist(\tilde{\rv G},\tilde{\rv G}')} \cdot \one\Set{\Abs{\cA(\tilde{\rv G})-p} < \frac{R\eta}{2} \pnull} 
    &\leq \sqrt{\E_{\omega,\cA} e^{2\eps\cdot\dist(\tilde{\rv G},\tilde{\rv G}')}} \sqrt{\E_{\omega,\cA} \one\Set{\Abs{\cA(\tilde{\rv G})-p} <\frac{R\eta}{2} \pnull}} \\
    &\leq \exp(\epsilon \eta n)\sqrt{\Pr_{\cA,\tilde\bbG(Q)} \Paren{\Abs{\cA(\tilde{\rv G})-p} < \frac{R\eta}{2} \pnull}} \,.
  \end{align*}
  Thus
  \[
    (1-\beta)^2 \leq \exp(\epsilon \eta n)\cdot \Pr_{\cA,\tilde\bbG(Q)} \Paren{\Abs{\cA(\tilde{\rv G})-p} < \frac{R\eta}{2} \pnull} \,.
  \]

  Then 
  \[
    \Pr_{\cA,\tilde\bbG(Q)} \Paren{\Abs{\cA(\tilde{\rv G})-p} < \frac{R\eta}{2} \pnull}
    \ge (1-\beta)^2 \cdot \exp(-\epsilon \eta n) \,.
  \]
  Since $\abs{p-\pnull}\geq R\eta \pnull$, the two events $\set{\hat{p} : \abs{\hat{p}-\pnull} < \frac{R\eta}{2} \pnull}$ and $\set{\hat{p} : \abs{\hat{p}-p} < \frac{R\eta}{2} \pnull}$ are disjoint, which implies
  \[
    \Pr_{\cA,\tilde\bbG(Q)} \Paren{\Abs{\cA(\tilde{\rv G})-p} < \frac{R\eta}{2} \pnull}
    \leq 1 - \Pr_{\cA,\tilde\bbG(Q)} \Paren{\Abs{\cA(\tilde{\rv G})-\pnull} < \frac{R\eta}{2} \pnull}
    \leq \beta \,.
  \]
  
  Therefore, we have $\beta\geq (1-\beta)^2 \exp(-\epsilon \eta n)$. 
  As result, we need to have 
  $\eta\geq \Omega\Paren{\frac{\log(\beta)}{\epsilon n}}$.
  Thus we have
  \begin{equation*}
      \abs{p-p'}\geq \Omega\Paren{\frac{R\log(\beta)p^{\circ}}{\epsilon n}}\,.
  \end{equation*}
  Since $p^{\circ}\geq p'$, it follows that
  \begin{equation*}
     \abs{p-p'}\geq \Omega\Paren{\frac{R\log(\beta)p'}{\epsilon n}}\,,  
  \end{equation*}
  which finishes the proof.
\end{proof}

%% file: preliminaries.tex
\section{Sum-of-squares background}
\label{sec:prelim}

\subsection{Sum-of-squares hierarchy}
In this paper, we use the sum-of-squares semidefinite programming hierarchy \cite{barak2014sum,sos2016note,raghavendra2018high} for both algorithm design and analysis.
The sum-of-squares proof-to-algorithm framework has been proven useful in many optimal or state-of-the-art results in algorithmic statistics~\cite{hopkins2018mixture,KSS18,pmlr-v65-potechin17a,hopkins2020mean}.
We provide here a brief introduction to pseudo-distributions, sum-of-squares proofs, and sum-of-squares algorithms.

\paragraph{Pseudo-distribution} 

We can represent a finitely supported probability distribution over $\R^n$ by its probability mass function $\mu\from \R^n \to \R$ such that $\mu \geq 0$ and $\sum_{x\in\supp(\mu)} \mu(x) = 1$.
We define pseudo-distributions as generalizations of such probability mass distributions by relaxing the constraint $\mu\ge 0$ to only require that $\mu$ passes certain low-degree non-negativity tests.

\begin{definition}[Pseudo-distribution]
  \label{def:pseudo-distribution}
  A \emph{level-$\ell$ pseudo-distribution} $\mu$ over $\R^n$ is a finitely supported function $\mu:\R^n \rightarrow \R$ such that $\sum_{x\in\supp(\mu)} \mu(x) = 1$ and $\sum_{x\in\supp(\mu)} \mu(x)f(x)^2 \geq 0$ for every polynomial $f$ of degree at most $\ell/2$.
\end{definition}

We can define the expectation of a pseudo-distribution in the same way as the expectation of a finitely supported probability distribution.

\begin{definition}[Pseudo-expectation]
  Given a pseudo-distribution $\mu$ over $\R^n$, we define the \emph{pseudo-expectation} of a function $f:\R^n\to\R$ by
  \begin{equation}
    \tE_\mu f \seteq \sum_{x\in\supp(\mu)} \mu(x) f(x) \,.
  \end{equation}
\end{definition}

The following definition formalizes what it means for a pseudo-distribution to satisfy a system of polynomial constraints.

\begin{definition}[Constrained pseudo-distributions]
  Let $\mu:\R^n\to\R$ be a level-$\ell$ pseudo-distribution over $\R^n$.
  Let $\cA = \{f_1\ge 0, \ldots, f_m\ge 0\}$ be a system of polynomial constraints.
  We say that \emph{$\mu$ satisfies $\cA$} at level $r$, denoted by $\mu \sdtstile{r}{} \cA$, if for every multiset $S\subseteq[m]$ and every sum-of-squares polynomial $h$ such that $\deg(h)+\sum_{i\in S}\max\set{\deg(f_i),r} \leq \ell$,
  \begin{equation}
    \label{eq:constrained-pseudo-distribution}
    \tE_{\mu} h \cdot \prod_{i\in S}f_i \ge 0 \,.
  \end{equation}
  We say $\mu$ satisfies $\cA$ and write $\mu \sdtstile{}{} \cA$ (without further specifying the degree) if $\mu \sdtstile{0}{} \cA$.
\end{definition}

We remark that if $\mu$ is an actual finitely supported probability distribution, then we have  $\mu\sdtstile{}{}\cA$ if and only if $\mu$ is supported on solutions to $\cA$.

\paragraph{Sum-of-squares proof} 

We introduce sum-of-squares proofs as the dual objects of pseudo-distributions, which can be used to reason about properties of pseudo-distributions.
We say a polynomial $p$ is a sum-of-squares polynomial if there exist polynomials $(q_i)$ such that $p = \sum_i q_i^2$.

\begin{definition}[Sum-of-squares proof]
  \label{def:sos-proof}
  A \emph{sum-of-squares proof} that a system of polynomial constraints $\cA = \{f_1\ge 0, \ldots, f_m\ge 0\}$ implies $q\ge0$ consists of sum-of-squares polynomials $(p_S)_{S\subseteq[m]}$ such that\footnote{Here we follow the convention that $\prod_{i\in S}f_i=1$ for $S=\emptyset$.}
  \[
    q = \sum_{\text{multiset } S\subseteq[m]} p_S \cdot \prod_{i\in S} f_i \,.
  \]
  If such a proof exists, we say that \(\cA\) \emph{(sos-)proves} \(q\ge 0\) within degree \(\ell\), denoted by $\mathcal{A}\sststile{\ell}{} q\geq 0$.
  In order to clarify the variables quantified by the proof, we often write \(\cA(x)\sststile{\ell}{x} q(x)\geq 0\).
  We say that the system \(\cA\) \emph{sos-refuted} within degree \(\ell\) if $\mathcal{A}\sststile{\ell}{} -1 \geq 0$.
  Otherwise, we say that the system is \emph{sos-consistent} up to degree \(\ell\), which also means that there exists a level-$\ell$ pseudo-distribution satisfying the system.
\end{definition}

The following lemma shows that sum-of-squares proofs allow us to deduce properties of pseudo-distributions that satisfy some constraints.
\begin{lemma}
  \label{lem:sos-soundness}
  Let $\mu$ be a pseudo-distribution, and let $\cA,\cB$ be systems of polynomial constraints.
  Suppose there exists a sum-of-squares proof $\cA \sststile{r'}{} \cB$.
  If $\mu \sdtstile{r}{} \cA$, then $\mu \sdtstile{r\cdot r' + r'}{} \cB$.
\end{lemma}

\paragraph{Sum-of-squares algorithm}

Given a system of polynomial constraints, the \emph{sum-of-squares algorithm} searches through the space of pseudo-distributions that satisfy this polynomial system by semideﬁnite programming.

Since semidefinite programing can only be solved approximately, we can only find pseudo-distributions that approximately satisfy a given polynomial system.
We say that a level-$\ell$ pseudo-distribution \emph{approximately satisfies} a polynomial system, if the inequalities in \cref{eq:constrained-pseudo-distribution} are satisfied up to an additive error of $2^{-n^\ell}\cdot \norm{h}\cdot\prod_{i\in S}\norm{f_i}$, where $\norm{\cdot}$ denotes the Euclidean norm\footnote{The choice of norm is not important here because the factor $2^{-n^\ell}$ swamps the effects of choosing another norm.} of the coefficients of a polynomial in the monomial basis.

\begin{theorem}[Sum-of-squares algorithm]
  \label{theorem:SOS_algorithm}
  There exists an $(n+ m)^{O(\ell)} $-time algorithm that, given any explicitly bounded\footnote{A system of polynomial constraints is \emph{explicitly bounded} if it contains a constraint of the form $\|x\|^2 \leq M$.} and satisfiable system\footnote{Here we assume that the bit complexity of the constraints in $\cA$ is $(n+m)^{O(1)}$.} $\cA$ of $m$ polynomial constraints in $n$ variables, outputs a level-$\ell$ pseudo-distribution that satisfies $\cA$ approximately.
\end{theorem}

\begin{remark}[Approximation error and bit complexity]
  \label{remark:sos-numerical-issue}  
  For a pseudo-distribution that only approximately satisfies a polynomial system, we can still use sum-of-squares proofs to reason about it in the same way as \cref{lem:sos-soundness}.
  In order for approximation errors not to amplify throughout reasoning, we need to ensure that the bit complexity of the coefficients in the sum-of-squares proof are polynomially bounded.  
\end{remark}

\subsection{Useful sum-of-squares lemmas}

\begin{lemma}
    \label{lem:preliminary-SOS-zero-one-range}
    \[
      \set{x^2=x} \sststile{2}{x} 0 \leq x \leq 1 \,.  
    \]
\end{lemma}

\begin{proof}
    The first inequality is trivial due to $\set{x^2=x} \sststile{2}{x} x = x^2 \geq 0$. For the second inequality, it follows that
    \begin{equation*}
        \set{x^2=x} \sststile{2}{x} x \leq \frac{x^2}{2} + \frac{1}{2} = \frac{x}{2} + \frac{1}{2} \,.
    \end{equation*}
    Rearranging the terms, we get
    \begin{equation*}
        \set{x^2=x} \sststile{2}{x} x \leq 1 \,.
    \end{equation*}
\end{proof}

\begin{lemma}
\label{lem:preliminary-SOS-set-union-bound}
    \begin{equation*}
        \set{x^2=x, y^2=y} \sststile{4}{x, y} 1-xy \leq (1-x)+(1-y) \,.
    \end{equation*}
\end{lemma}

\begin{proof}
    By \cref{lem:preliminary-SOS-zero-one-range}, it follows that
    \begin{equation*}
        \set{x^2=x, y^2=y} \sststile{2}{x, y} 0 \leq x, y \leq 1 \,.
    \end{equation*}
    Therefore, we have
    \begin{align*}
        \set{x^2=x, y^2=y}
        & \sststile{4}{x, y} (1-y)(1-x) \geq 0 \\
        & \sststile{4}{x, y} 1-x-y \geq -xy \\
        & \sststile{4}{x, y} 2-x-y \geq 1-xy \,.
    \end{align*}
\end{proof}

\begin{lemma}
\label{lem:preliminary-SOS-abs-to-square}
Given constant $C$, we have
    \begin{equation*}
        \set{-C \leq x \leq C} \sststile{2}{x} x^2 \leq C^2  \,.
    \end{equation*}
\end{lemma}

\begin{proof}
    \begin{align*}
        \set{-C \leq x \leq C}
        & \sststile{2}{x} (C-x)(C+x) \geq 0 \\
        & \sststile{2}{x} C^2-x^2 \geq 0 \\
        & \sststile{2}{x} C^2 \geq x^2 \,.
    \end{align*}
\end{proof}

\begin{lemma}
\label{lem:preliminary-SOS-square-to-abs}
Given constant $C$, we have
    \begin{equation*}
        \set{x^2 \leq C^2} \sststile{2}{x} -C \leq x \leq C  \,.
    \end{equation*}
\end{lemma}

\begin{proof}
    For the first inequality, we have
    \begin{equation*}
        \set{x^2 \leq C^2}
        \sststile{2}{x}
        x \geq -\frac{x^2}{2C} - \frac{C}{2}
        \geq -\frac{C^2}{2C} - \frac{C}{2}
        = - C \,.
    \end{equation*}
    For the second inequality, we have
    \begin{equation*}
        \set{x^2 \leq C^2}
        \sststile{2}{x}
        x \leq \frac{x^2}{2C} + \frac{C}{2}
        \leq \frac{C^2}{2C} + \frac{C}{2}
        = C \,.
    \end{equation*}
\end{proof}

%% file: concentration-inequalities.tex
\section{Concentration inequalities}
\label{sec:concentration-inequalities}

\begin{lemma}[Average degree concentration]
    \label{lem:average_degree_concentration}
    Let $\Qnull$ be an $n$-by-$n$ edge connection probability matrix and let $\dnull \seteq d(\Qnull)$.
    Let $\rv{A}\sim\bbG(n,\Qnull)$.
    Then for any $\delta\in(0,1)$,
    \[
        \Pr\Paren{\Abs{d(\rv A) - \dnull} \ge \delta \dnull} \le 2\exp\Paren{-\frac{\delta^2 n \dnull}{6}} \,,
    \]
\end{lemma}
\begin{proof}
    Let $\mu \seteq \E \sum_{i<j} \rv A_{ij} = \sum_{i<j} p_{ij}$.
    Using Chernoff bound, for $\delta\in(0,1)$,
    \begin{align*}
        \Pr\Paren{\Abs{\sum_{i<j} \rv A_{ij} - \mu} \ge \delta\mu} &\le 2\exp\Paren{-\frac{\delta^2\mu}{3}} \,, \\
        \Pr\Paren{\Abs{d(\rv A) - \dnull} \ge \delta \dnull} &\le 2\exp\Paren{-\frac{\delta^2 n \dnull}{6}} \,.
    \end{align*}
\end{proof}

\begin{lemma}[Degree distribution]
    \label{lem:degree_distribution}
    Let $\Qnull$ be an $n$-by-$n$ edge connection probability matrix.
    Let $d$ be a parameter such that $d\ge5$ and $\normi{\Qnull} \le d/n$.
    Then for every $t\in[2e^2, \log n]$, an inhomogeneous random graph $\bbG(n,\Qnull)$ has at least $e^{-t}n$ nodes with degree at least $td$ with probability at most $\exp\paren{-te^{-t}nd/4}$.
\end{lemma}
\begin{proof}
    Let $\rv m_{k}$ denote the number of nodes with degree at least $k$ in $\bbG(n,\Qnull)$.
    Then for every $\gamma\in[0,1]$, 
    \begin{align*}
        \Pr\Paren{\rv m_{td} \ge \gamma n}
        &\le {n \choose \gamma n} {\gamma n^2 \choose \gamma n td/2} \Paren{\frac{d}{n}}^{\gamma n td/2} \\
        &\le \Paren{\frac{e}{\gamma}}^{\gamma n} \Paren{\frac{2e}{t}}^{\gamma n td/2} \\
        &= \exp\Paren{-\gamma n\Paren{\frac{td}{2}\log\frac{t}{2e} - \log\frac{e}{\gamma}}} \,.
    \end{align*}
    Plugging in $\gamma=e^{-t}$ gives
    \[
        \Pr\Paren{\rv m_{td} \ge e^{-t}n}
        \le \exp\Paren{-te^{-t}n\Paren{\frac{d}{2}\log\frac{t}{2e} - 1-1/t}} \,.
    \]
    For $t\in[2e^2, \log n]$ and $d\ge5$,
    \[
        \Pr\Paren{\rv m_{td} \ge e^{-t}n}
        \le \exp\Paren{-te^{-t}n\Paren{\frac{d}{2} - \frac{5}{4}}}
        \le \exp\Paren{-te^{-t}nd/4} \,.
    \]
\end{proof}

\begin{lemma}[Degree pruning]
    \label{lem:degree_pruning}
    Let $\Qnull$ be an $n$-by-$n$ edge connection probability matrix.
    Let $d$ be a parameter such that $d\ge5$ and $\normi{\Qnull} \le d/n$.
    Then with probability at least $1-n^{1-d/4}$, an inhomogeneous graph $\bbG(n,\Qnull)$ has the following property. For all $t\in[2e^2, \log n]$, the number of edges incident to nodes with degree at least $td$ is at most $2te^{-t} nd$;
\end{lemma}
\begin{proof}
    Let $\rv m_{k}$ denote the number of nodes with degree at least $k$ in $\bbG(n,\Qnull)$.
    By \cref{lem:degree_distribution}, for any $t\in[2e^2, \log n]$ and $d\ge5$,
    \[
        \Pr\Paren{\rv m_{td} \ge e^{-t}n}
        \le \exp\Paren{-te^{-t}nd/4}
        \le n^{-d/4} \,.
    \]
    Applying union bound, the event that $m_{k} \le e^{-k/d}n$ for any integer $k \in [2e^2d, (\log n)d]$ happens with probability at least $1-n^{1-d/4}$.
    We condition our following analysis on this event.
    
    Fix a $t\in[2e^2, \log n]$.
    The number of edges incident to nodes with degree at least $td$ is at most
    \begin{align*}
        \sum_{i=0} (t+i+1)d \cdot e^{-(t+i)}n 
        = nd \sum_{i=t} (i+1) e^{-i}
        = nd e^{-t} \Paren{\frac{e}{e-1}t + \frac{e^2}{(e-1)^2}}
        \leq 2te^{-t} nd \,.
    \end{align*}
\end{proof}

\begin{lemma}[Degree-truncated subgraph]
    \label{lem:degree_truncated_subgraph}
    Let $\Qnull$ be an $n$-by-$n$ edge connection probability matrix and let $\dnull\seteq d(\Qnull)$.
    Let $d$ be a parameter such that $d\ge5$ and $\normi{\Qnull} \le d/n$.
    For $\delta\in(0,1)$, an inhomogeneous graph $\rv{A}\sim\bbG(n,\Qnull)$ has the following property with probability at least $1-n^{1-d/4}-\exp(-\delta^2n\dnull/6)$.
    For every $t\in[2e^2, \log n]$, $A$ contains an $n$-node subgraph $\tilde{A}$ of  such that 
    \begin{itemize}
        \item $(\tilde{A}\one)_i \le t d$ for any $i\in[n]$;
        \item $(1-\delta)\dnull - 4te^{-t}d \le d(\tilde{A}) \le (1+\delta)\dnull$.
    \end{itemize}
\end{lemma}
\begin{proof}
    By \cref{lem:average_degree_concentration} and \cref{lem:degree_pruning}, $\rv A \sim \bbG(n,\Qnull)$ has the following two properties with probability at least $1-n^{1-d/4}-\exp(-\delta^2n\dnull/6)$.
    \begin{itemize}
        \item $\abs{d(A) - \dnull} \le \delta \dnull$.
        \item For all $t\in[2e^2, \log n]$, the number of edges incident to nodes with degree at least $td$ is at most $2te^{-t} nd$.
    \end{itemize}
    Consider a graph $A$ with the above two properties.
    Fix a $t\in[2e^2, \log n]$.
    By removing at most $2te^{-t} nd$ edges from $A$, we can obtain a graph $\tilde A$  such that the maximum degree of $\tilde A$ is at most $td$. 
    Moreover, 
    \begin{align*}
        & d(\tilde{A}) \ge d(A) - 4te^{-t}d \ge (1-\delta)\dnull - 4te^{-t}d \,, \\
        & d(\tilde{A}) \le d(A) \le (1+\delta)\dnull \,.
    \end{align*}
\end{proof}

\begin{lemma}[Spectral bound~\cite{bandeira2016sharp}]
    \label{lem:spectral_bound_with_degree_pruning}
    Let $\rv{A} \sim \bbG(n,p_0)$ and suppose $np_0\ge5$.
    Then with probability at least $1-n^{-\Omega(1)}$,
    \[
        \Normop{\rv A - p_0(\one\one^\top-\Id)} \le 
        O\Paren{\sqrt{np_0\log n}} \,.
    \]
\end{lemma}

%% file: main.bbl
\newcommand{\etalchar}[1]{$^{#1}$}
\providecommand{\bysame}{\leavevmode\hbox to3em{\hrulefill}\thinspace}
\providecommand{\MR}{\relax\ifhmode\unskip\space\fi MR }
\providecommand{\MRhref}[2]{%
  \href{http://www.ams.org/mathscinet-getitem?mr=#1}{#2}
}
\providecommand{\href}[2]{#2}
\begin{thebibliography}{CCAd{\etalchar{+}}23}

\bibitem[AJ06]{adell2006exact}
Jos{\'e}~A Adell and Pedro Jodr{\'a}, \emph{Exact kolmogorov and total
  variation distances between some familiar discrete distributions}, Journal of
  Inequalities and Applications \textbf{2006} (2006), 1--8.

\bibitem[AJK{\etalchar{+}}22]{acharya2022robust}
Jayadev Acharya, Ayush Jain, Gautam Kamath, Ananda~Theertha Suresh, and Huanyu
  Zhang, \emph{Robust estimation for random graphs}, Conference on Learning
  Theory, PMLR, 2022, pp.~130--166.

\bibitem[AKT{\etalchar{+}}23]{alabi2023privately}
Daniel Alabi, Pravesh~K Kothari, Pranay Tankala, Prayaag Venkat, and Fred
  Zhang, \emph{Privately estimating a gaussian: Efficient, robust, and
  optimal}, Proceedings of the 55th Annual ACM Symposium on Theory of
  Computing, 2023, pp.~483--496.

\bibitem[AUZ23]{asi2023robustness}
Hilal Asi, Jonathan Ullman, and Lydia Zakynthinou, \emph{From robustness to
  privacy and back}, International Conference on Machine Learning, PMLR, 2023,
  pp.~1121--1146.

\bibitem[BBDS13]{blocki2013differentially}
Jeremiah Blocki, Avrim Blum, Anupam Datta, and Or~Sheffet, \emph{Differentially
  private data analysis of social networks via restricted sensitivity},
  Proceedings of the 4th conference on Innovations in Theoretical Computer
  Science, 2013, pp.~87--96.

\bibitem[BC17]{borgs2017graphons}
Christian Borgs and Jennifer Chayes, \emph{Graphons: A nonparametric method to
  model, estimate, and design algorithms for massive networks}, Proceedings of
  the 2017 ACM Conference on Economics and Computation, 2017, pp.~665--672.

\bibitem[BCS15]{borgs2015private}
Christian Borgs, Jennifer~T. Chayes, and Adam Smith, \emph{Private graphon
  estimation for sparse graphs}, 2015.

\bibitem[BCSZ18]{borgs2018revealing}
Christian Borgs, Jennifer Chayes, Adam Smith, and Ilias Zadik, \emph{Revealing
  network structure, confidentially: Improved rates for node-private graphon
  estimation}, 2018.

\bibitem[BJR07]{bollobas2007phase}
B{\'e}la Bollob{\'a}s, Svante Janson, and Oliver Riordan, \emph{The phase
  transition in inhomogeneous random graphs}, Random Structures \& Algorithms
  \textbf{31} (2007), no.~1, 3--122.

\bibitem[BS14]{barak2014sum}
Boaz Barak and David Steurer, \emph{Sum-of-squares proofs and the quest toward
  optimal algorithms}, Proceedings of the {I}nternational {C}ongress of
  {M}athematicians---{S}eoul 2014. {V}ol. {IV}, Kyung Moon Sa, Seoul, 2014,
  pp.~509--533. \MR{3727623}

\bibitem[BS16]{sos2016note}
\bysame, \emph{Proofs, beliefs, and algorithms through the lens of
  sum-of-squares}, lecture notes, 2016.

\bibitem[BvH16]{bandeira2016sharp}
Afonso~S. Bandeira and Ramon van Handel, \emph{Sharp nonasymptotic bounds on
  the norm of random matrices with independent entries}, Ann. Probab.
  \textbf{44} (2016), no.~4, 2479--2506. \MR{3531673}

\bibitem[CCAd{\etalchar{+}}23]{chen2023private}
Hongjie Chen, Vincent Cohen-Addad, Tommaso d’Orsi, Alessandro Epasto, Jacob
  Imola, David Steurer, and Stefan Tiegel, \emph{Private estimation algorithms
  for stochastic block models and mixture models}, Advances in Neural
  Information Processing Systems \textbf{36} (2023), 68134--68183.

\bibitem[CDd{\etalchar{+}}24]{chen2024private}
Hongjie Chen, Jingqiu Ding, Tommaso d'Orsi, Yiding Hua, Chih-Hung Liu, and
  David Steurer, \emph{Private graphon estimation via sum-of-squares}, arXiv
  preprint arXiv:2403.12213 (2024).

\bibitem[DL09]{dwork2009differential}
Cynthia Dwork and Jing Lei, \emph{Differential privacy and robust statistics},
  Proceedings of the forty-first annual ACM symposium on Theory of computing,
  2009, pp.~371--380.

\bibitem[DMNS06]{dwork2006calibrating}
Cynthia Dwork, Frank McSherry, Kobbi Nissim, and Adam Smith, \emph{Calibrating
  noise to sensitivity in private data analysis}, Theory of Cryptography: Third
  Theory of Cryptography Conference, TCC 2006, New York, NY, USA, March 4-7,
  2006. Proceedings 3, Springer, 2006, pp.~265--284.

\bibitem[DR14]{dwork2014algorithmic}
Cynthia Dwork and Aaron Roth, \emph{The algorithmic foundations of differential
  privacy}, Foundations and Trends{\textregistered} in Theoretical Computer
  Science \textbf{9} (2014), no.~3--4, 211--407.

\bibitem[DSSU17]{dwork2017exposed}
Cynthia Dwork, Adam Smith, Thomas Steinke, and Jonathan Ullman, \emph{Exposed!
  a survey of attacks on private data}, Annual Review of Statistics and Its
  Application \textbf{4} (2017), 61--84.

\bibitem[FKP{\etalchar{+}}19]{fleming2019semialgebraic}
Noah Fleming, Pravesh Kothari, Toniann Pitassi, et~al., \emph{Semialgebraic
  proofs and efficient algorithm design}, Foundations and
  Trends{\textregistered} in Theoretical Computer Science \textbf{14} (2019),
  no.~1-2, 1--221.

\bibitem[GH22]{Georgiev22RobustPrivacy}
Kristian Georgiev and Samuel Hopkins, \emph{Privacy induces robustness:
  Information-computation gaps and sparse mean estimation}, Advances in Neural
  Information Processing Systems (S.~Koyejo, S.~Mohamed, A.~Agarwal,
  D.~Belgrave, K.~Cho, and A.~Oh, eds.), vol.~35, Curran Associates, Inc.,
  2022, pp.~6829--6842.

\bibitem[HKM22]{hopkins2022efficient}
Samuel~B Hopkins, Gautam Kamath, and Mahbod Majid, \emph{Efficient mean
  estimation with pure differential privacy via a sum-of-squares exponential
  mechanism}, Proceedings of the 54th Annual ACM SIGACT Symposium on Theory of
  Computing, 2022, pp.~1406--1417.

\bibitem[HKMN23]{Hopkins2023Privacy}
Samuel~B. Hopkins, Gautam Kamath, Mahbod Majid, and Shyam Narayanan,
  \emph{Robustness implies privacy in statistical estimation}, Proceedings of
  the 55th Annual ACM Symposium on Theory of Computing (New York, NY, USA),
  STOC 2023, Association for Computing Machinery, 2023, p.~497–506.

\bibitem[HL18]{hopkins2018mixture}
Samuel~B. Hopkins and Jerry Li, \emph{Mixture models, robustness, and sum of
  squares proofs}, Proceedings of the 50th Annual ACM SIGACT Symposium on
  Theory of Computing (New York, NY, USA), STOC 2018, Association for Computing
  Machinery, 2018, p.~1021–1034.

\bibitem[HLL83]{holland1983stochastic}
Paul~W Holland, Kathryn~Blackmond Laskey, and Samuel Leinhardt,
  \emph{Stochastic blockmodels: First steps}, Social networks \textbf{5}
  (1983), no.~2, 109--137.

\bibitem[Hop20]{hopkins2020mean}
Samuel~B Hopkins, \emph{Mean estimation with sub-gaussian rates in polynomial
  time}, The Annals of Statistics \textbf{48} (2020), no.~2, 1193--1213.

\bibitem[HRH02]{hoff2002latent}
Peter~D Hoff, Adrian~E Raftery, and Mark~S Handcock, \emph{Latent space
  approaches to social network analysis}, Journal of the american Statistical
  association \textbf{97} (2002), no.~460, 1090--1098.

\bibitem[KMV22]{kothari22Private}
Pravesh Kothari, Pasin Manurangsi, and Ameya Velingker, \emph{Private robust
  estimation by stabilizing convex relaxations}, Proceedings of Thirty Fifth
  Conference on Learning Theory (Po-Ling Loh and Maxim Raginsky, eds.),
  Proceedings of Machine Learning Research, vol. 178, PMLR, 02--05 Jul 2022,
  pp.~723--777.

\bibitem[KNRS13]{kasiviswanathan2013analyzing}
Shiva~Prasad Kasiviswanathan, Kobbi Nissim, Sofya Raskhodnikova, and Adam
  Smith, \emph{Analyzing graphs with node differential privacy}, Theory of
  Cryptography: 10th Theory of Cryptography Conference, TCC 2013, Tokyo, Japan,
  March 3-6, 2013. Proceedings, Springer, 2013, pp.~457--476.

\bibitem[KRSY11]{karwa2011private}
Vishesh Karwa, Sofya Raskhodnikova, Adam Smith, and Grigory Yaroslavtsev,
  \emph{Private analysis of graph structure}, Proceedings of the VLDB Endowment
  \textbf{4} (2011), no.~11, 1146--1157.

\bibitem[KSS18]{KSS18}
Pravesh~K Kothari, Jacob Steinhardt, and David Steurer, \emph{Robust moment
  estimation and improved clustering via sum of squares}, Proceedings of the
  50th Annual ACM SIGACT Symposium on Theory of Computing, 2018,
  pp.~1035--1046.

\bibitem[LKKO21]{liu2021robust}
Xiyang Liu, Weihao Kong, Sham Kakade, and Sewoong Oh, \emph{Robust and
  differentially private mean estimation}, Advances in neural information
  processing systems \textbf{34} (2021), 3887--3901.

\bibitem[LKO22]{liu2022differential}
Xiyang Liu, Weihao Kong, and Sewoong Oh, \emph{Differential privacy and robust
  statistics in high dimensions}, Conference on Learning Theory, PMLR, 2022,
  pp.~1167--1246.

\bibitem[MT07]{mcsherry2007mechanism}
Frank McSherry and Kunal Talwar, \emph{Mechanism design via differential
  privacy}, 48th Annual IEEE Symposium on Foundations of Computer Science
  (FOCS'07), IEEE, 2007, pp.~94--103.

\bibitem[NRS07]{nissim2007smooth}
Kobbi Nissim, Sofya Raskhodnikova, and Adam Smith, \emph{Smooth sensitivity and
  sampling in private data analysis}, Proceedings of the thirty-ninth annual
  ACM symposium on Theory of computing, 2007, pp.~75--84.

\bibitem[NS09]{narayanan2009anonymizing}
Arvind Narayanan and Vitaly Shmatikov, \emph{De-anonymizing social networks},
  2009 30th IEEE symposium on security and privacy, IEEE, 2009, pp.~173--187.

\bibitem[PS17]{pmlr-v65-potechin17a}
Aaron Potechin and David Steurer, \emph{Exact tensor completion with
  sum-of-squares}, Proceedings of the 2017 Conference on Learning Theory, 2017.

\bibitem[RSS18]{raghavendra2018high}
Prasad Raghavendra, Tselil Schramm, and David Steurer, \emph{High dimensional
  estimation via sum-of-squares proofs}, Proceedings of the International
  Congress of Mathematicians: Rio de Janeiro 2018, World Scientific, 2018,
  pp.~3389--3423.

\bibitem[SU19]{ullman2019efficiently}
Adam Sealfon and Jonathan Ullman, \emph{Efficiently estimating erdos-renyi
  graphs with node differential privacy}, Advances in Neural Information
  Processing Systems \textbf{32} (2019).

\end{thebibliography}
